\documentclass[letterpaper]{article}
\usepackage[margin=1in]{geometry}

\title{Subadditive Load Balancing}

\author{ 
{\bf Kiyohito Nagano} \\
Gunma University\\
\texttt{k-nagano@gunma-u.ac.jp} \vspace{2mm}\\
\and 
{\bf Akihiro Kishimoto}  \\
IBM Research, Ireland\\ 
\texttt{akihirok@ie.ibm.com} 
}

\usepackage{graphicx} 
\usepackage{hyperref}
\usepackage{amsmath,amssymb,amsthm,psfrag}
\usepackage{tabularx,supertabular,booktabs}
\usepackage{comment}

\def\hat{\widehat}

\def\R{\mathbb{R}}

\def\emptyset{\varnothing}

\def\apath{{\textsf{Path}}}
\def\amst{{\textsf{MST}}}
\def\modumin{{\textsf{MMin}}}

\newtheorem{theorem}{Theorem}
\newtheorem{prop}[theorem]{Proposition}
\newtheorem{lemma}[theorem]{Lemma}

\begin{document}

\maketitle

\begin{abstract}
Set function optimization is essential in AI and machine learning. 
We focus on a subadditive set function that generalizes submodularity, and examine the subadditivity of non-submodular functions. 
We also deal with a minimax subadditive load balancing problem, and present a modularization-minimization algorithm that theoretically guarantees a worst-case approximation factor.
In addition, we give a lower bound computation technique for the problem. 
We apply these methods to the multi-robot routing problem for an empirical performance evaluation. 
\end{abstract}

\section{Introduction}

A set function is regarded as a discrete function on the vertices of $n$-dimensional hypercube $\{ 0,\, 1 \} ^n$. Many combinatorial problems
arising in machine learning
%AI and machine learning
are reduced to set function optimization. 
In particular, submodular set functions are a fundamental tool. 
%in combinatorial optimization, machine learning and other related fields,
%for 
%real-world 
%applications. 
%and has been recognized as an interesting subject of research. 
%Minimization problems and maximization problems of submodular set functions are elemental problems.
%, and and many combinatorial problems arising in machine learning can be reduced to submodular function optimization. 
For example, submodular 
%set function
% minimization 
optimization 
 has been used to perform clustering \cite{NJB05,NKI10}, image segmentation \cite{SK10,JB11} and feature selection \cite{Bac13}, and applied to influence maximization \cite{KKT03}, sensor placement \cite{GKS05}, and text summarization \cite{LB10}.
%Kishi: moved later
%This paper focuses on a subadditive set function that generalizes submodularity in a simple manner. 
%More specifically, we deal with the minimax subadditive load balancing problem, and apply this problem to the multi-robot routing (MRR) problem with the minimax team objective.

A set function $f$ is a real-valued function defined on subsets of a finite set $V = \{ 1, \ldots,\, n \}$. 
The domain of $f$ is the power set of $V$, denoted by $2^V = \{ S : S \subseteq V\}$.
A set function $f:2^V \to \mathbb{R}$ is submodular if $f(S)+f(T) \ge f(S \cup T)+f(S \cap T)$ for all $S, T \subseteq V$.
The submodular set function is known to be a discrete counterpart of the convex function \cite{Lov83}. 
Similarly to convex functions, submodular functions can be exactly minimized in polynomial time \cite{GLS88,Sch00,IFF01}. 
However, in many practical settings, submodular function minimization becomes a very difficult task even when only one simple additional constraint is introduced \cite{SF08,IN09}.
%\cite{SF08,IN09,GKTW09}. 
%
Minimax submodular load balancing (SMLB) is NP-hard. Svitkina and Fleischer \cite{SF08} presented a sampling-based $O(\sqrt{n \ln n }) $-approximation algorithm.
Wei \textit{et al.}~\cite{WIWBB15} tackled submodular partitioning problems, including minimax SMLB, and presented 
a majorization-minimization algorithm that ensures a theoretical worst-case approximation factor. 
%However, their analysis of the factor heavily relies on the curvatures of the submodular set functions \cite{Von10,IJB13}. 
Their analysis heavily relies on the curvatures of the submodular functions \cite{Von10,IJB13,SVW15}. 
%Moreover, Wei \textit{et al.} \cite{WIWBB15} showed a number of machine learning applications of submodular partitioning problems. 

% *** Kishi: I perhaps will either remove this paragraph or add a bit sentence to it  ***
% We focus on a subadditive set function that generalizes aforementioned submodularity in a simple manner. 
% More specifically, we deal with the minimax subadditive load balancing problem, and apply this problem to the multi-robot routing (MRR) problem with the minimax team objective.
% *** emphasize this is a new domain research subadditive load balancing has not been used before ***

A set function $g:2^V \to \mathbb{R}$ is subadditive if $g(S)+g(T) \ge g(S \cup T)$ for all $S, T \subseteq V$. Nonnegative submodularity immediately leads to nonnegative subadditivity  
% There are several fundamental functions that are not submodular but subadditive, e.g., 
% the facility location function \cite{GS04} and the minimum spanning tree (MST) function. 
whose
%Non-submodular
optimization is important in machine learning \cite{BBKT17}.  
% Despite a simple generalization, little work exists on general subadditive optimization \cite{Fei09}.
% 
% 
%%% 
A fractionally subadditive (or XOS) set function, which is a special case of a subadditive set function and a generalization of a submodular set function, is well studied in the context of combinatorial auction \cite{DNS10} and learnability and sketchability of set functions \cite{BH18,BDFKNR12,FV16}. 
Despite a simple generalization of submodularity, to the best of our knowledge, little work exists on general subadditive optimization \cite{Fei09}.
% It is hard to say that general subadditive optimization algorithms and applications of the subadditive optimization have been sufficiently studied.
Therefore, theoretical properties and potential applications of such general subadditive set optimization problems have not been revealed yet.

We, therefore, first examine the subadditivity of fundamental non-submodular functions, including
the facility location function \cite{GS04} and the minimum spanning tree (MST) function.
% In addition, we show that interpolation of a submodular set function with unknown function values is computed with a subadditive set function in a simple way.
We show that a subadditive set function simplifies computing interpolation of a submodular set function with unknown function values. 
The interpolation is related to submodular function approximation \cite{GHIM09}, but our approach is different, since the algorithm of Goemans \textit{et al.}~\cite{GHIM09} is not always easy to implement.  

We then consider the minimax subadditive load balancing (SALB) problem as an important generalization of SMLB.
% We present the modularization-minimization algorithm that theoretically guarantees a worst-case approximation factor for nondecreasing subadditive cases.
% The presented algorithm is a variant of the majorization-minimization algorithm \cite{WIWBB15} for SMLB. 
% However, unlike the case of SMLB \cite{WIWBB15}, our new analysis of the approximation factor heavily depends on the curvatures of the subadditive set functions. 
As a variant of the majorization-minimization algorithm \cite{WIWBB15} for SMLB, 
we present the modularization-minimization algorithm that ensures a theoretical worst-case approximation factor. 
Our analysis reveals the difference and similarity between submodular and subadditive set functions
in terms of tractability. 
% However, unlike the case of SMLB \cite{WIWBB15}, our new analysis of the approximation factor heavily depends on the curvatures of the subadditive set functions. 
%
%
While the approximability of SALB implies a tractable aspect of subadditivity,  
%Contrastingly, we also
we prove intractability of curvature computation of a subadditive set function. 
Thus, we introduce a concept of a pseudo-curvature that is relatively tractable. 
In addition, we present a method for computing a lower bound for SALB in some special cases, including SMLB. 

Finally, we discuss that the SALB problem with the MST functions is related to multi-robot routing (MRR) problem with the minimax team objective \cite{MRR05}, 
%We apply our approach to solve MRR for an empirical performance evaluation.
and perform an empirical evaluation of our approach. 
Besides, the iterative procedure in the modularization-minimization algorithm attempts to improve a solution found so far. 
We empirically evaluate the effectiveness of the iterative procedure starting with different initial solutions computed by other existing 
%approximation
 algorithms for MRR. 
% Our results empirically show the importance of the choice of the initial solutions. 
% On other other hand, the iterative procedure does not currently contribute to improving the worst-case approximation factor. 
% Therefore, our insights would open up an opportunity as future work regarding elucidating the theoretical and empirical behaviors of the algorithm with respect to the initial solutions. 

\begin{comment}
The remainder of the paper is organized as follows. 
In Section 2, we provide 
% basic definitions of set functions, 
examples of subadditive functions, and 
the definition of the subadditive load balancing.
In Section 3, we describe the modularization-minimization algorithm for solving SALB, and give its approximation factor. 
In Section 4 we prove intractable aspects of subadditivity.
% Section 4 explains the application of SALB to MRR. 
% Finally, we show some empirical results of computational experiments in Section 5,
% and give concluding remarks in Section 6. 
Section 5 explains the application of SALB to MRR. 
Finally, we show empirical results in MRR in Section 6,
and give concluding remarks in Section 7. 
\end{comment}

\section{Subadditive functions and subadditive load balancing}

We first give basic definitions, and then
%We first give basic definitions of set functions and examples of subadditive functions. 
%(for details on the theory of submodular functions, see \cite{Fuj05,Sch03}).
%We then
define the subadditive load balancing problem.

\subsection{Set functions, and subadditive functions} 

Let $V = [n] := \{ 1,\, \ldots ,\, n \}$ be a given set of $n$ elements, and $g: 2^V \to \mathbb{R}$ be a real-valued function defined on all the subset of $V$. 
Such a function $g$ is called a \textit{set function} with a ground set $V$. 
A set function $g: 2^V \to \mathbb{R}$ is called \textit{subadditive} if
% \begin{align} 
% \label{eq:subad}
$g(S)+g(T) \ge g(S \cup T) , \, \forall S, T \subseteq V$.
%\end{align}
A set function $g: 2^V \to \mathbb{R}$ is called \textit{submodular} if
% \begin{align} 
% \label{eq:submo}
$g(S)+g(T) \ge g(S \cup T) + g(S \cap T), \, \forall S, T \subseteq V$, 
% \end{align}
%A set function $g$ is called \textit{supermodular} if $-g$ is submodular.
and 
%A set function
is called \textit{modular} 
%(\textit{linear}, or \textit{additive})
 if 
$g(S)+g(T) = g(S \cup T) + g(S \cap T), \, \forall S, T \subseteq V$.
%it always satisfies \eqref{eq:submo} with equality.
A set function is called \textit{nonnegative} if $0 \le g(S)$ for any $S\subseteq V$, 
%A set function is called 
\textit{nondecreasing} if $g(S) \le g(T)$ for any $S, T\subseteq V$ with $S \subseteq T $, 
and 
\textit{normalized} if $g(\emptyset ) =0$. 
Trivially, a nonnegative submodular function is a nonnegative subadditive function.
Thus subadditivity generalizes submodularity in a simple manner. 
For an $n$-dimensional vector $\boldsymbol{z} = (z_i ) _{i \in V } \in \mathbb{R} ^n$ and $S \subseteq V$, we denote $z(S) = \sum _{i \in S } z_i $. A set function $z: 2^V \to \mathbb{R}$ corresponding to the vector $\boldsymbol{z} \in \mathbb{R} ^n$ is a modular function with $z(\emptyset ) = 0$.
In the rest of the paper, we basically denote a subadditive set function as $g$, and a submodular set function as $f$.

% For an $n$-dimensional vector $\boldsymbol{a} \in \mathbb{R} ^n$ with components $a_i$, $i \in V$, and a subset $S \subseteq V$, we denote $a(S) = \sum _{i \in S } a_i $. 
% For convenience, we let $a(\emptyset ) = 0$. 
% A set function $a: 2^V \to \mathbb{R}$ corresponding to the vector $\boldsymbol{a} $ is a modular function.

\subsection{Examples of subadditive set functions} \label{sec:example}

We examine the subadditivity of some non-submodular functions.
%We give some examples of set functions that are not submodular but subadditive.
% \vspace{-3mm}
The minimum spanning tree function plays an important role in multi-robot routing problems in \S{}5. 
Another example, a subadditive interpolation of a submodular set function, is given in \S{}\ref{sec:si}.
%the supplementary material.

% \subsubsection{Minimum spanning tree function}
% \label{sec:mst}

\subsubsection{Examples in combinatorial optimization}

\paragraph{Minimum spanning tree function.} 
The \textit{minimum spanning tree function} is a canonical example of the subadditive set function. 
Let $r$ be a root node and $V = \{ 1,\, \ldots ,\, n \}$ be a set of other nodes. 
We are given a distance $d(i,j) \ge 0 $ for any $i,\, j \in \widetilde{V}:= \{ r \} \cup V$. 
Suppose that $d: \widetilde{V} \times \widetilde{V} \to \mathbb{R}$ is symmetric and satisfies triangle inequalities. 
For any subset $S \subseteq V$, a minimum spanning tree (MST) w.r.t. $\widetilde{S} := \{ r \} \cup S $
is a spanning tree w.r.t. $\widetilde{S}$ that minimizes the sum of edge distances.
For any $S \subseteq V$, let $MST(S) $ be a sum of edge distances of MST w.r.t. $\widetilde{S}$. 
We call $MST : 2^V \to \mathbb{R}$ a minimum spanning tree function on $V$ with root $r$. 

\begin{lemma} \label{lem:mst}
A minimum spanning tree function $MST : 2^V \to \mathbb{R}$ is nonnegative and subadditive.
\end{lemma}
% \vspace{-4mm}
% \begin{proof} 
% By definition, nonnegativity is trivial.
% For subsets $S, T\subseteq V$, 
% let $E_S $ be the edge set of MST w.r.t. $\widetilde{S}$ and let $E_T $ be the edge set of MST w.r.t. $\widetilde{T}$. 
% The graph $(S \cup T \cup \{ r \} ,\, E_S \cup E_T ) $ with a node set $S \cup T \cup \{ r \}$ and an edge set $E_S \cup E_T$ are connected. 
% Thus we have 
% $
% MST(S) + MST(T) = \sum _{e \in E_S \cup E_T} d(e) \ge MST(S \cup T),
% $
% which shows the subadditivity of $MST$. 
% \vspace{-2mm}
% \end{proof} 

The function $MST : 2^V \to \mathbb{R}$ is not always nondecreasing and/or submodular.
For the function $MST$ in Figure \ref{fig:mst} (a),  
$MST (\{ 1, 3 \}) =10$ and $MST (\{ 1, 2, 3 \}) = 9$. Thus,  $MST$ is not nondecreasing. 
For the function $MST$ in Figure \ref{fig:mst} (b), 
$MST (\{ 1 \}) =5$, $MST (\{ 1, 2 \}) = MST (\{ 1, 3 \}) = 6$, and $MST (\{ 1, 2, 3 \}) = 9$.
%Therefore,  $MST (\{ 1, 2 \}) + MST (\{ 1, 3 \}) < MST (\{ 1, 2, 3 \}) + MST (\{ 1 \})$. 
Thus,  $MST$ is not submodular. \vspace{0mm}

\begin{figure}[htbp]
        \begin{center}
        \includegraphics[scale=0.45]{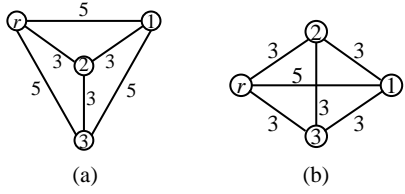}
        \ \vspace{-2mm}
        \caption{Minimum spanning tree functions}
        \label{fig:mst}
        \end{center}
        \ \vspace{-7mm}
\end{figure}

Let us consider an MST function with nonnegative node weights.
Given an MST function $MST : 2^V \to \mathbb{R}$ and a nonnegative uniform node weight $\beta \ge 0$, 
the function $MST ^{\beta } : 2^V \to \mathbb{R}$ defined by
\begin{align} \label{eq:mstbeta}
MST^{\beta} (S) = MST (S) + \beta |S| \ (\forall S \subseteq V)
\end{align} 
is also subadditive. 
The functions $MST$ and $MST^{\beta}$ will be used in the computational experiments 
on multi-robot routing problems in Section \ref{sec:er}.

% \subsubsection{Facility location function} 
\paragraph{Facility location function.} 
The facility location function \cite{GS04} is another example.
% of the subadditive set function. 
Given a finite set $V=\{1,2,\ldots ,\, n\} $ of customers and a finite set $F$ of possible locations for the facilities, a certain service is provided by connecting customers to opened facilities.
Opening facility $j \in F$ incurs a fixed cost $o_j \ge 0$, and connecting customer $i \in V$ to facility $j \in F $ incurs a fixed cost $c_{ij} \ge 0$. 
For any subset $S \subseteq V$ of the customers, let $FL(S) $ be the minimum cost of providing the service only to $S$. We call $FL : 2^V \to \mathbb{R}$ a facility location function. 

Let us see an example of $FL$ in Figure \ref{fig:fl} with $F =\{ a,\, b \} $ and $V = \{ 1, 2, 3 \}$. 
We have, e.g., $FL (\{ 2 \}) = o_a + c_{2a} = 2$, $FL (\{ 1,\, 2 \}) = o_a + c_{1a} + c_{2a} = 3$, 
and $FL (\{ 1,\, 2 , \, 3 \}) = o_a + c_{1a} + c_{2a} + o_b + c_{3b}= 5$.

% The following property of $FL$ is satisfied.
%The set function $FL$ satisfies the following property. 

\begin{lemma} \label{lem:fl}
A facility location function $FL : 2^V \to \mathbb{R}$ is nondecreasing and subadditive.
\end{lemma}

As pointed out in \cite{GS04}, 
%the function 
$FL : 2^V \to \mathbb{R}$ is not necessarily submodular.
%
% Consider the function $FL$ in Figure \ref{fig:fl} with $F =\{ a,\, b \} $ and $V = \{ 1, 2, 3 \}$. 
% We have $FL (\{ 2 \}) = 2$, $FL (\{ 1, 2 \}) = FL (\{ 2, 3 \}) = 3$, and $FL (\{ 1, 2, 3 \}) = 5$.
% Therefore,  $FL (\{ 1, 2 \}) + FL (\{ 2, 3 \}) < FL (\{ 1, 2, 3 \}) + FL (\{ 2 \})$. 
% Thus,  $FL$ is not submodular. \vspace{1mm} 
In the case of Figure \ref{fig:fl}, 
we have 
%$FL (\{ 2 \}) = 2$, $FL (\{ 1, 2 \}) = FL (\{ 2, 3 \}) = 3$, and $FL (\{ 1, 2, 3 \}) = 5$.
$FL (\{ 1, 2 \}) + FL (\{ 2, 3 \}) < FL (\{ 1, 2, 3 \}) + FL (\{ 2 \})$.
%, which indicates the non-submodularity of $FL$. 

\begin{figure}[htbp]
        \begin{center}
        \includegraphics[scale=0.5]{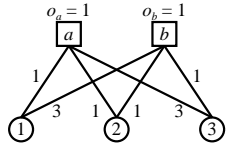}
        %\ \vspace{0mm}
        \caption{Facility location function}
        \label{fig:fl}
        \end{center}
        \ \vspace{-8mm}
\end{figure}

%\begin{comment}
\subsubsection{Interpolation of a submodular set function} \label{sec:si}
% \paragraph{Interpolation of a submodular set function.} 

%\paragraph{Subadditive interpolation of a submodular set function.} 
Let $f: 2^V \to \R$ be a nondecreasing \textit{submodular} set function with $f(\emptyset ) = 0$. 
Assume that we have only a part of the function values of $f$. That is, 
for a given collection $\mathcal{C} = \{ C_1 ,\, C_2,\, \ldots ,\, C_m \} \subseteq 2^V$, 
the function values $f(C_i ) = f_i $ are known for each $i= 1, \ldots ,\, m $ and the value $f(S) $ are unknown for any $S\in 2^V \setminus \mathcal{C}$. 
The objective here is to build a set function $g: 2^V \to \R$ that approximates $f$. 
We introduce a general, simple method to construct a subadditive interpolating set function $g _{\mathcal{C}}$, which is computationally tractable. 
% $g_{\mathcal{C}}$ of $f$. 
We utilize the ideas of the polymatroid \cite{Edm70} and the Lov\'asz extension \cite{Lov83}. 

In some applications, evaluating function values of $f$ is computationally expensive (see, e.g., \cite{Les07}). 
By appropriately setting the collection $\mathcal{C}$, our interpolation method transforms a complicated submodular optimization problem into a simple subadditive optimization problem. 
\vspace{-2mm}

\paragraph{Lov\'asz extension.}
The \textit{polymatroid} $\mathrm{P} (f )  = 
\{ \boldsymbol{z} \in \mathbb{R} ^n : \mbox{$z(S) \le f(S)$} \ (\forall S \subseteq V) \} 
\cap \mathbb{R} _{\ge 0} ^n$
%, \ z _i \ge 0 \ (\forall i \in V)\} $
 is a bounded polyhedron, where $\mathbb{R} _{\ge 0}$ is the set of nonnegative real values.
The \textit{Lov\'asz extension} $\hat{f}: \mathbb{R} _{\ge 0} ^n \to \mathbb{R}$ is defined by 
$
\hat{f} (\boldsymbol{x}) = \max _{\boldsymbol{z} \in \mathrm{P} (f ) } 
\langle \boldsymbol{x} , \boldsymbol{z} \rangle
\ \ (\forall \boldsymbol{x} \in \mathbb{R} _{\ge 0} ^n )
$, where $\langle \boldsymbol{x} , \boldsymbol{z} \rangle = \sum _{i \in V} x_i z_i $. 
The function $\hat{f} $ is a natural continuous extension of $f$ since 
$\hat{f} (\boldsymbol{I} _S ) = f(S)$ holds, $\forall S \subseteq V$, where $\boldsymbol{I} _S \in \{0,\, 1\} ^n$ is the characteristic vector of $S$. 
\vspace{-2mm}

\paragraph{Construction of the interpolating function.} 
For
 an \textit{imitated} polymatroid 
\[
\mathrm{P}_{\mathcal{C}} (f )  = 
\{ \boldsymbol{z} \in \mathbb{R} ^n : \mbox{$z(C_i ) \le f (C_i )$} \ (\forall i = 1, \ldots ,\, m) \} \cap \mathbb{R} _{\ge 0} ^n 
, 
\] 
and an \textit{imitated} Lov\'asz extension 
$
\hat{f} _{\mathcal{C}} (\boldsymbol{x}) = \max _{\boldsymbol{z} \in \mathrm{P}_{\mathcal{C}} (f ) } 
\langle \boldsymbol{x} , \boldsymbol{z} \rangle
\ \ (\forall \boldsymbol{x} \in \mathbb{R} _{\ge 0} ^n )$, 
we define the set function $g _{\mathcal{C}} : 2^V \to \mathbb{R}$ as 
%Then, the set function $g _{\mathcal{C}} : 2^V \to \mathbb{R}$ defined by 
$
g_{\mathcal{C}} (S) = \hat{f} _{\mathcal{C}} (\boldsymbol{I} _S )  \ (\forall S \subseteq V) 
$. The following lemma guarantees that $g_{\mathcal{C}} $ is a natural subadditive extension of $f$ and 
$g_{\mathcal{C}} $ is computationally tractable.
\vspace{1mm}

\begin{lemma} \label{lem:si}
The set function $g_{\mathcal{C}} : 2^V \to \mathbb{R}$ satisfies \\
\begin{tabular}{rl}
\textbf{(i)} &$g_{\mathcal{C}} (S) \ge f(S) $, $\forall S \subseteq V$, \\
\textbf{(ii)} &$g_{\mathcal{C}} (C_i )= f(C_i ) $, $\forall i= 1, \ldots ,\, m $, \\
\textbf{(iii)} &$g_{\mathcal{C}}$ is a nondecreasing \textit{subadditive} set function, and \\
\textbf{(iv)} &the value $g_{\mathcal{C}} (S)$ can be computed in polynomial time in $n$ and $m$ for any given $S \subseteq V$.
\end{tabular}
% The set function $g_{\mathcal{C}} : 2^V \to \mathbb{R}$ satisfies \textbf{(i)} $g_{\mathcal{C}} (S) \ge f(S) $, $\forall S \subseteq V$, \textbf{(ii)} $g_{\mathcal{C}} (C_i )= f(C_i ) $, $\forall i= 1, \ldots ,\, m $, \textbf{(iii)} $g_{\mathcal{C}}$ is a nondecreasing \textit{subadditive} set function, and \textbf{(iv)} the value $g_{\mathcal{C}} (S)$ can be computed in polynomial time in $n$ and $m$ for any given $S \subseteq V$.
\end{lemma}
\vspace{-0mm}

\begin{proof}
\textbf{Proof of (i).} 
By definition, we have $\mathrm{P} (f ) \subseteq \mathrm{P}_{\mathcal{C}} (f ) $. 
Therefore, 
\begin{align*}
f(S) = \hat{f} (\boldsymbol{I} _S ) 
 = 
\max\limits _{\boldsymbol{z} \in \mathrm{P} (f ) } 
\langle \boldsymbol{I} _S, \boldsymbol{z} \rangle  \le 
\max\limits _{\boldsymbol{z} \in \mathrm{P}_{\mathcal{C}} (f ) } 
\langle \boldsymbol{I} _S , \boldsymbol{z} \rangle 
= \hat{f} _{\mathcal{C}} (\boldsymbol{I} _S ) = g_{\mathcal{C}} (S) , 
\end{align*}
for all $S \subseteq V$. \vspace{1mm}

%\hfill $\square$\vspace{0mm}
\noindent
\textbf{Proof of (ii).} 
By Lemma 3 (i), we have $f(C_i ) \le g_{\mathcal{C}} (C_i )$, $\forall i= 1, \ldots ,\, m $. 
By the definition of $g_{\mathcal{C}} $ and $\mathrm{P}_{\mathcal{C}} (f )$, we have 
\[
g_{\mathcal{C}} (C_i ) 
= 
\max\limits _{\boldsymbol{z} \in \mathrm{P} _{\mathcal{C}} (f ) } 
\langle \boldsymbol{I} _{C_i}, \boldsymbol{z} \rangle
=
\max\limits _{\boldsymbol{z} \in \mathrm{P} _{\mathcal{C}} (f ) } 
z(C_i ) 
\le f(C_i ),
\] 
for all $i= 1, \ldots ,\, m $.
Thus, we obtain $f(C_i ) = g_{\mathcal{C}} (C_i )$, for all $i= 1, \ldots ,\, m $. 
\vspace{1mm}

%\hfill $\square$\vspace{0mm}
\noindent
\textbf{Proof of (iii).} The nondecreasing property of $g_{\mathcal{C}}$ follows from $\mathrm{P}_{\mathcal{C}} (f ) \subseteq \mathbb{R} _{\ge 0} ^n$. 
We have 
$g_{\mathcal{C}} (S\cup T) = 
\max _{\boldsymbol{z} \in \mathrm{P}_{\mathcal{C}} (f ) } 
\langle \boldsymbol{I} _{S\cup T} , \boldsymbol{z} \rangle
\le 
\max _{\boldsymbol{z} \in \mathrm{P}_{\mathcal{C}} (f ) } 
\langle \boldsymbol{I} _{S} + \boldsymbol{I} _{T} , \boldsymbol{z} \rangle \le 
\max _{\boldsymbol{z} \in \mathrm{P}_{\mathcal{C}} (f ) } 
\langle \boldsymbol{I} _{S} , \boldsymbol{z} \rangle
+
\max _{\boldsymbol{z} \in \mathrm{P}_{\mathcal{C}} (f ) } 
\langle \boldsymbol{I} _{T} , \boldsymbol{z} \rangle
= g(S) + g(T)$, for any $S,\, T \subseteq V$.  
% \begin{align*}
% g_{\mathcal{C}} (S\cup T) = 
% \max\limits _{\boldsymbol{z} \in \mathrm{P}_{\mathcal{C}} (f ) } 
% \langle \boldsymbol{I} _{S\cup T} , \boldsymbol{z} \rangle
% &\le 
% \max\limits _{\boldsymbol{z} \in \mathrm{P}_{\mathcal{C}} (f ) } 
% \langle \boldsymbol{I} _{S} + \boldsymbol{I} _{T} , \boldsymbol{z} \rangle
% \\
% & \le 
% \max\limits _{\boldsymbol{z} \in \mathrm{P}_{\mathcal{C}} (f ) } 
% \langle \boldsymbol{I} _{S} , \boldsymbol{z} \rangle
% +
% \max\limits _{\boldsymbol{z} \in \mathrm{P}_{\mathcal{C}} (f ) } 
% \langle \boldsymbol{I} _{T} , \boldsymbol{z} \rangle
% = g(S) + g(T). 
% \end{align*}
Thus, $g_{\mathcal{C}}$ is subadditive. 
%\hfill $\square$\vspace{0mm}
\vspace{1mm}

\noindent
\textbf{Proof of (iv).} 
The polyhedron $\mathrm{P}_{\mathcal{C} } (f ) $ 
is bounded and it is determined by $n+m$ linear inequalities. Thus, the linear programming problem 
$\max\limits _{\boldsymbol{z} \in \mathrm{P}_{\mathcal{C}} (f ) } 
\langle \boldsymbol{I} _S , \boldsymbol{z} \rangle$ 
can be solved  in polynomial time in $n$ and $m$. 
%\hfill $\square$\vspace{3mm}
\end{proof}

The function $g_{\mathcal{C}}$ is not necessarily submodular.
Consider the case where $V=\{ 1,2, 3\}$, $f (S) = (7 - |S| ) |S| \ (\forall S \subseteq V)$, and $\mathcal{C} = 2^V \setminus V$. We have $g_{\mathcal{C}} (\{ 1 \} ) =6$, 
$g_{\mathcal{C}} (\{ 1 ,\, 2 \} ) = g_{\mathcal{C}} (\{ 1 ,\, 3 \} ) = 10$, and $g_{\mathcal{C}} (\{ 1 ,\, 2 ,\, 3\} ) = 15$. Therefore, $g_{\mathcal{C}}$ is not submodular. 
% \end{comment}
Figure \ref{fig:pfpcf} illustrates a polymatroid $\mathrm{P} (f )$ and an imitated polymatroid $\mathrm{P}_{\mathcal{C}} (f )$ in the case of this example. 

\begin{figure}[htbp]
        \begin{center}
        \includegraphics[scale=0.4]{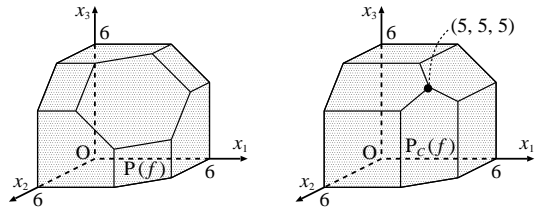}
        \ \vspace{-2mm}
        \caption{Polymatroid $\mathrm{P} (f )$ and imitated polymatroid $\mathrm{P}_{\mathcal{C}} (f )$}
        \label{fig:pfpcf}
        \end{center}
        \ \vspace{-7mm}
\end{figure}

\subsection{Subadditive load balancing}

We define the submodular load balancing (SMLB) and the subadditive load balancing (SALB) problems.
$\mathcal{S} = (S_1, \, \ldots ,\, S_m)$ is an $m$-partition of $V =\{ 1, \ldots ,\, n \}$ if 
$S_1 \cup \cdots \cup S_m = V$ and $S_j \cap S_{j'} = \emptyset $  for $1\le j < j' \le m$ (some $S_j$ can be empty). 
Suppose that set functions $f_1,\, \ldots ,\, f_m : 2^V \to \mathbb{R}$ are normalized, nonnegative, and submodular, and set functions $g_1,\, \ldots ,\, g_m : 2^V \to \mathbb{R}$ are normalized, nonnegative, and subadditive. 
The SMLB problem is defined as 
%%\ \vspace{2mm}\\
% \begin{align} \label{eq:smlb}
% \mbox{min} \ \ \max\limits _{j=1, \ldots ,\, m} f_j (S_j) 
% \ \ \ \mbox{s.\,t.} \ \ \  
% \mathcal{S} = (S_1, \, \ldots ,\, S_m) \mbox{ is an $m$-partition of }V .
% 
\begin{align} \label{eq:smlb}
\hspace{-2mm} 
\begin{array}{ll}
\mbox{min} & \max\limits _{j=1, \ldots ,\, m} f_j (S_j) \\
\mbox{s.\,t.} & \mathcal{S} = (S_1, \, \ldots ,\, S_m) \mbox{ is an $m$-partition of }V .
\end{array}
\end{align}
We say that SMLB is nondecreasing if $f_1,\, \ldots ,\, f_m$ are nondecreasing. 
Approximation algorithms and heuristics are proposed for the nondecreasing SMLB \cite{SF08,WIWBB15}.  
%%
%%\\vspace{2mm} 
On the other hand, we mainly deal with SALB which has not been performed in-depth analysis before. 
SALB has a slightly different objective function: 
\begin{align} \label{eq:salb}
% \mbox{min} \ \ \max\limits _{j=1, \ldots ,\, m} g_j (S_j) 
% \ \ \ \mbox{s.\,t.} \ \ \ \mathcal{S} = (S_1, \, \ldots ,\, S_m) \mbox{ is an $m$-partition of }V .
\hspace{-2mm} 
\begin{array}{ll}
\mbox{min} & \max\limits _{j=1, \ldots ,\, m} g_j (S_j) \\
\mbox{s.\,t.} & \mathcal{S} = (S_1, \, \ldots ,\, S_m) \mbox{ is an $m$-partition of }V .
\end{array}
\end{align}
%SALB is nondecreasing if $g_1,\, \ldots ,\, g_m$ are nondecreasing.
SALB is nondecreasing for nondecreasing $g_1,\, \ldots ,\, g_m$. 
%, which requires a more careful analysis.  

%*** Kishi: perhaps either move these sentences to sections 3 and 4 or remove them ***. 
% In Section \ref{sec:algo}, we present a modularization-minimization algorithm that theoretically guarantees a worst-case approximation factor for the nondecreasing SALB.
% The proposed algorithm is a natural generalization of the majorization-minimization algorithm for the nondecreasing SMLB \cite{WIWBB15}.  
% In Section \ref{sec:mrr}, we show that SALB can be utilized to solve the multi-robot routing problem with the minimax team objective.  

\section{Algorithms for SALB}
\label{sec:algo}

SALB is NP-hard due to the hardness of SMLB. 
We, therefore, consider some approaches to find approximate solutions of SALB,
and begin with the greedy method. \vspace{2mm}\\
% \vspace{-2mm}
%
%We begin with the greedy method for SALB. \ \vspace{2mm}\\
%The greedy method \textsf{Greedy} is straightforward but does not empirically yield solutions of good quality compared to the approach to which we give new theoretical analysis (see Section 5 for the experimental results). \ \vspace{2mm}\\
\begin{tabularx}{165mm}{l} \toprule[1pt] 
\hspace{-2mm}
\textbf{Algorithm} $\mbox{\textsf{Greedy}}(g_1, \ldots , g_m )$ 
\vspace{0.5mm}\\
\hline 
\end{tabularx}\vspace{1mm}\\
%%
%\hspace{-2mm}
\begin{tabularx}{165mm}{lX} 
\hspace{-1.5mm} \textbf{0}: \hspace{-4.5mm} & Set $S_j := \emptyset $,  $\forall j \in [m] 
$, and $U:=V$.\\
%\textbf{0}:& Set $S_j := \emptyset $,  $\forall j \in [m] $, and $U:=V$.\\
\hspace{-1.5mm} \textbf{1}: \hspace{-4.5mm} & \textbf{While} $U \neq \emptyset$ \textbf{do} \\
& 
\ \ \ Choose $i_j \in \arg \min\limits _{i \in R } g_j (S_j \cup \{ i \} )$, $\forall j \in [m]$ 
\\
&
\ \ \ Choose $j^* \in \arg \min\limits _{j \in [m] } g_j (S_j \cup \{ i _j \} )$ 
\\
&
\ \ \ Set $S_{j^*} := S_{j^*} \cup \{ i _{j^*} \} $ and $U := U \setminus \{ i _{j^*} \} $
\\
\hspace{-1.5mm} \textbf{2}:  \hspace{-4.5mm} & Output $\mathcal{S} = (S_1, \, \ldots ,\, S_m) $. \\
%%\vspace{1mm}\\
\bottomrule[1pt]
\end{tabularx}\vspace{3mm}\\
\textsf{Greedy} is straightforward but does not empirically yield solutions of good quality
 (see Section 5). 
We introduce a nontrivial approach and give new theoretical analyses.

% \vspace{5mm}\\
%
We say that an algorithm $\mathcal{A}$ for the minimization problem (P) achieves an approximation factor of $\gamma \ge 1$ or $\mathcal{A}$ is a $\gamma $-approximation algorithm 
if $(OPT \le) \ APP \le \gamma \cdot OPT$ is satisfied for any instance of the problem (P), where $OPT$ is the optimal value of the problem and $APP$ is the objective function value of the approximate solution obtained by $\mathcal{A}$.
It is known to be difficult to give a theoretically good approximation algorithm for SALB in a sense. 
%The result of \cite{SF08} shows that there
There is no polynomial-time approximation algorithm for the nondecreasing SMLB (and SALB) with an approximation factor of $o (\sqrt{n / \ln n})$ \cite{SF08}. 
Therefore, it is important to see the tractability of load balancing problems by using measures different from $n=|V|$. 

As a nontrivial approach to SMLB,
Wei \textit{et al.} \cite{WIWBB15} proposed 
a majorization-minimization algorithm and gave a worst-case approximation factor for the nondecreasing submodular case.
Its approximation factor depends on the curvatures of the submodular set functions. 

We replace $f_1, \ldots, f_m$ in the majorization-minimization algorithm of 
%Wei \textit{et al.}~
\cite{WIWBB15} with $g_1, \ldots, g_m$, which we call the \emph{modularization-minimization} algorithm \textsf{MMin} (see \S{}\ref{sec:mmin}). This replacement leads to the following notable differences and similarities described in this and next sections: 
%% and we attempt to clarify the difference and similarity between submodular functions and subadditive functions from the viewpoint of tractability.  
%% Differences and similarities we will see are summarized as follows. 
\vspace{-2mm}
\begin{itemize}
\item %As a variant of the majorization-minimization algorithm for SMLB, we describe a modularization-minimization algorithm \textsf{MMin} for SALB in \S{}\ref{sec:mmin}. 
%  However, in the case of SALB,
  Unlike SMLB, 
  a \textit{majorizing} approximation modular set function for SALB cannot be constructed 
due to the non-submodular structure.
\item 
As is the case for SMLB, our analysis uses the curvatures of the subadditive set functions. 
The worst-case approximation factor of \textsf{MMin} for the nondecreasing SALB in \S{}\ref{sec:ana} is a generalization of the result of \cite{WIWBB15} for the nondecreasing SMLB.
%As is the case for SMLB, our analysis uses the curvatures of the subadditive set functions. However, the worst-case approximation factor of \textsf{MMin} for the nondecreasing SALB in \S{}\ref{sec:ana} is a generalization of the result of \cite{WIWBB15} for the nondecreasing SMLB.
%%
\item Unlike for a submodular set function, the curvature computation is not easy for a subadditive set function (see \S{}\ref{sec:int}).
 Note that the curvature computation is not necessarily required because the algorithm does not use its value.  \vspace{-2mm}
\end{itemize}

The approximation guarantee including the curvatures of set functions may be unstable or useless, 
due to its difficulty of computing the actual value of the approximation factor. 
In order to resolve this issue, we present a method to compute a lower bound of the optimal solution for some important cases (see \S{}\ref{sec:lb}).
%The method itself can be regarded as a variant of the algorithm \textsf{MMin}.  
%\S{}\ref{sec:lb_sm} deals with the nondecreasing SMLB, and \S{}\ref{sec:lb_mst} deals with the SALB with the minimum spanning tree functions defined in \S{}\ref{sec:example}. 

% \newpage

\subsection{The modularization minimization algorithm}
\label{sec:mmin}

We describe the modularization-minimization algorithm \textsf{MMin} for SALB.
%, and give an approximation factor for the nondecreasing case.
\vspace{-3mm}

%\subsubsection*{Framework of the algorithm}

\paragraph{Framework of the algorithm.} 
%The algorithm iteratively \textsf{MMin} updates the $m$-partition $\mathcal{S}$.
\textsf{MMin} iteratively updates the $m$-partition $\mathcal{S}$. 
%% Let us see the update operation of each iteration. 
%% Suppose that we are given a tentative $m$-partition $\mathcal{S} ' = (S' _1, \, \ldots ,\, S'  _m) $ for SALB. 
%% Then we construct modular approximation function $M_j : 2^V \to \mathbb{R}$ of $g_j $ at $S' _j$ for each $j = 1,\ldots ,\, m$, and obtain 
%% a modified $m$-partition $\mathcal{S} '' = (S_1 '' , \, \ldots ,\, S_m '' )$
%% which is an optimal solution to the following modular load balancing (M-LB) problem:
Given a tentative $m$-partition $\mathcal{S} ' = (S' _1, \, \ldots ,\, S'  _m) $ for SALB,
the update operation of each iteration
constructs modular approximation function $M_j : 2^V \to \mathbb{R}$ of $g_j $ at $S' _j$ for each $j = 1,\ldots ,\, m$, and obtains 
a modified $m$-partition $\mathcal{S} '' = (S_1 '' , \, \ldots ,\, S_m '' )$. 
$\mathcal{S} ''$ is an optimal solution to the following modular load balancing (M-LB) problem:
\begin{align} \label{eq:mlb}
% \mbox{min} \ \ \ \max\limits _{j=1, \ldots ,\, m} M_j (S_j) 
% \ \ \ \mbox{s.\,t.} \ \ \ \mathcal{S} = (S_1, \, \ldots ,\, S_m) \mbox{ is an $m$-partition of 
% }V .
\hspace{-2mm} 
\begin{array}{ll}
\mbox{min} & \max\limits _{j=1, \ldots ,\, m} M_j (S_j) \\
\mbox{s.\,t.} & \mathcal{S} = (S_1, \, \ldots ,\, S_m) \mbox{ is an $m$-partition of }V .
\end{array}
\end{align}
%\vspace{-6mm}
%%
%%
\paragraph{Construction of approximation functions.} 
Given a subadditive set function $g : 2^V \to \mathbb{R}$ and a subset $S' \subseteq V$, 
we have to construct a modular approximation set function $M$ of $g$ at $S'$ in order to deal with the problem \eqref{eq:mlb}. 
If $g$ is \textit{submodular}, a \textit{majorization set function} $M$ satisfying $g(S) \le M(S) ,\, \forall S$ and $g(S') = M(S') $ 
can be constructed in a simple way \cite{WIWBB15}.\footnote{
In the \textit{submodular} case \cite{WIWBB15}, the functions $M^1 $ and $M^2$ defined by, for all $S \subseteq V$, 
$M^1 (S) = g(S' ) + \sum\limits _{i \in S \setminus S' } g( i \, | \, S' )
 - \sum\limits _{i \in S' \setminus S } g( i \, | \, V \setminus \{ i \} )$,   
$M^2 (S) = g(S' ) + \sum\limits _{i \in S \setminus S' } g( i \, | \, \emptyset )
 - \sum\limits _{i \in S' \setminus S } g( i \, | \, S' \setminus \{ i \} )$ are both majorization set functions. These functions are not necessarily majorizing in the subadditive case. 
}
In contrast, in the subadditive case, it would be difficult to construct a majorization set function. Thus, for example, we 
% as a modular approximation function of $g$ at $S'$, e.g., we can
 use an intuitively natural modular set function $M $ defined by 
%, e.g.,
\begin{align}
\displaystyle 
\mbox{$
M(S) = g(S' ) + \sum\limits _{i \in S \setminus S' } g( i \, | \, S' )
 - \sum\limits _{i \in S' \setminus S } g( i \, | \, S' \setminus \{ i \} ) \ \ \ (S \subseteq V), 
$}
% \\
% \hspace{-3mm} 
% \begin{array}{l}
% \displaystyle M(S) = g(S' ) + \sum\limits _{i \in S \setminus S' } g( i \, | \, S' )\\
% \phantom{M(S) = \displaystyle g(S' ) } 
% \displaystyle  - \sum\limits _{i \in S' \setminus S } g( i \, | \, S' \setminus \{ i \} ) \ \ (S \subseteq V),
% \end{array}
\end{align}
where $g(i \, | \, S ) = g(S \cup \{ i \}) - g(S)$ for $S \subseteq V$ and $i \notin S $.
If $g$ is nondecreasing, the marginal cost $g(i \, | \, S )$ is nonnegative for all $S \subseteq V$ and $i \notin S $.

In computing lower bounds (see \S{} \ref{sec:lb}), we will give an alternative way of constructing the approximation modular set function $M$ which is a \textit{minorization set function}.
The minorization set function $M$ approximating the function $g$ around the subset $S'$ satisfies 
\[
\mbox{
$g(S) \ge M(S) ,\, \forall S$ \ and \ $g(S') = M(S') $.
}
\]
We use an \textit{approximate minorization set function} for some important special cases. 
\vspace{-1mm}

%\subsection*{Algorithm description} 

\paragraph{Algorithm description.} 
\textsf{MMin} is described below.\vspace{2mm}\\
\begin{tabularx}{165mm}{l} \toprule[1pt] 
\hspace{-2mm}
\textbf{Algorithm} $\mbox{\textsf{MMin}}(g_1, \ldots , g_m )$ 
% \end{tabularx}\\
% \begin{tabularx}{85mm}{lX}
% \textit{ Input:} & Subadditive set functions $g_1, \ldots , g_m$.
% (given by a value-oracle).
% \\
% \textit{ Output:} & An $m$-partition $\mathcal{S} = (S_1  , \, \ldots ,\, S_m ) $ of $V$. \\
% \textit{ Objective:} & $\max \{ g_1 (S_1) ,\, \ldots ,\, g_m (S_m) \} $ 
% is minimized.
\vspace{0.5mm}\\
\hline 
\end{tabularx}\vspace{1mm}\\
\begin{tabularx}{165mm}{lX} 
\hspace{-1.5mm} \textbf{0}: \hspace{-5.5mm} 
& Find an initial $m$-partition $\mathcal{S} ^{(0)} = (S_1 ^{(0)}, \, \ldots ,\, S_m ^{(0)})$ of $V$. Set $k:= 1$.
%\vspace{1mm}
\\
\hspace{-1.5mm} \textbf{1}: \hspace{-5.5mm} 
& Construct a modular approximation function $M _j ^{(k)}$ of $g_j $ at $S^{(k-1 )} _j $, $\forall j \in [m] = \{ 1,\ldots ,\, m \}$. 
%\vspace{1mm}
\\
\hspace{-1.5mm} \textbf{2}: \hspace{-5.5mm} 
& Let $\mathcal{S} ^{(k)} = (S_1 ^{(k)}, \, \ldots ,\, S_m ^{(k)})$ be an $m$-partition $\mathcal{S}$ that minimizes $\max\limits _{j=1, \ldots ,\, m} M_j ^{(k) }(S_j)$.  
\vspace{1mm}\\
\hspace{-1.5mm} \textbf{3}: \hspace{-5.5mm} 
& \textbf{If} $\mathcal{S} ^{(k)} = \mathcal{S} ^{(k-1)} $\\
& \ \ \ \textbf{then} output $\mathcal{S} := \mathcal{S} ^{(k)}$, \\
& \ \ \ \textbf{else} set $k:=k+1$ and go to Step 1. 
\\
\bottomrule[1pt]
\end{tabularx}\vspace{4mm}\\
\textsf{MMin} needs to find an initial partition in Step 0, and solve the M-LB problem \eqref{eq:mlb} in Step 2. 
We describe the methods to do so in the following part. 
In addition, we give a simple approximation bound of \textsf{MMin} for SALB with a mild assumption.
%\vspace{-2mm}

% \subsection*{Finding an initial partition} 

\paragraph{Finding an initial partition.} 
\textsf{MMin} can start with an arbitrary $m$-partition of $V$. 
In order to obtain an approximation factor for SALB, 
we consider the M-LB problem 
\begin{align} \label{eq:mlb0}
% \begin{array}{l}
% \mbox{min} \ \ \max\limits _{j=1, \ldots ,\, m} M ^{(0)} _{j} (S_j) 
% \ \ \ 
% \mbox{s.\,t.} \ \ \ \mathcal{S} = (S_1, \, \ldots ,\, S_m) \mbox{ is an $m$-partition of }V ,\\ 
% \mbox{where } 
%  M ^{(0)} _{j} (S) = \sum _{i \in S} g_j (\{ i \}) \ (S \subseteq V), \ \forall j \in [m],
% \end{array}
\hspace{-2mm}
\begin{array}{l}
\mbox{min} \ \  \max\limits _{j=1, \ldots ,\, m} M ^{(0)} _{j} (S_j) \\
\mbox{s.\,t.} \ \ \ \, \mathcal{S} = (S_1, \, \ldots ,\, S_m) \mbox{ is an $m$-partition of }V ,\\ 
\mbox{where } \displaystyle M ^{(0)} _{j} (S) = \sum _{i \in S} g_j (\{ i \}) \ (S \subseteq V),
\end{array}
\end{align}
and let $\mathcal{S} ^{(0)}$ be an optimal (or approximately optimal) partition of problem \eqref{eq:mlb0}. 
% \vspace{-3mm}

%\subsection*{Solving the modular load balancing} 

\paragraph{Solving the modular load balancing.} 
Each modular set function $M_j $ in problem \eqref{eq:mlb} 
is represented as $M_j (S) = b_j + \sum _{i \in S} c_{ij} \ (S \subseteq V)$. 
Therefore, by using a standard MIP (mixed integer programming) formulation technique, problem \eqref{eq:mlb} becomes
\begin{align} \label{eq:milp}
% \begin{array}{ll}
% \mbox{min} \ \ \ y \ \ \ 
% \mbox{s.\,t.} 
% & 
% \sum\limits _{j \in [m] } x_{ij}= 1,\ \forall i \in V, \ \ \ \ 
% b_j + \sum\limits _{i \in S} c_{ij} x_{ij}  \le y ,\ \forall j \in [m] ,\\
% & 
% \displaystyle
% x_{ij} \in \{ 0,\, 1 \}, \forall i \in V , \  \forall j \in [m] , \ \ \ \ y \in \mathbb{R}. 
% \end{array}
\begin{array}{ll}
\mbox{min} & y \\
\mbox{s.\,t.} 
& 
\displaystyle \sum _{j \in [m] } x_{ij}= 1,\ \forall i \in V, 
\\
& 
\displaystyle b_j + \sum _{i \in S} c_{ij} x_{ij}  \le y ,\ \forall j \in [m] ,\\
& 
\displaystyle
x_{ij} \in \{ 0,\, 1 \}, \forall i \in V , \  \forall j \in [m] , \ \ y \in \mathbb{R}. 
\end{array}
\end{align}
An optimal solution to problem \eqref{eq:milp} can be found via MIP solvers such as IBM ILOG CPLEX. 
An LP-based 2-approximation 
% (polynomial-time)
 algorithm of \cite{LST90} for the unrelated parallel machines scheduling problem can also be used for problem \eqref{eq:milp}. However,
with the algorithm of Lenstra et al., the approximation factor of 2 can be preserved only when $b_j \ge 0 \ (\forall j \in [m])$ and $c_{ij} \ge 0 \ (\forall i \in V,\, \forall j \in [m])$.

\paragraph{A simple bound of the algorithm.} 

For 
%a subadditive set function
 $g: 2^V \to \mathbb{R}$, we say that $g$ satisfies 
a \textit{singleton-minimal} (SMinimal) property if 
%\begin{align*}
$g(\{ i \}) \le g(S \cup \{ i \} ) ,  \mbox{ for all } i \in V \mbox{ and } S \subseteq V \setminus \{ i \} $.
%\end{align*}
Any nondecreasing set function and the MST function in \S{}\ref{sec:example}, which is not necessarily nondecreasing, satisfy the SMinimal property.  
%The MST function in \S{}\ref{sec:example}, which is not necessarily nondecreasing, also satisfies this property.  
We give a simple approximation bound of \textsf{MMin} for SALB. 
% in a case where each $g_j$ satisfies the  SMinimal property.

% Let $\mathcal{S} ^{(0)} = (S_1 ^{(0)}, \, \ldots ,\, S_m ^{(0)})$ be an 
% optimal partition of \eqref{eq:mlb0}, and let
% $\mathcal{S} ^* = (S_1 ^* ,\, \ldots ,\, S_m ^* )$ is an optimal partition of a general SALB.
% Then, we have the following lemma.

\begin{prop} \label{th:simpleb}
The algorithm \textsf{MMin} for SALB achieves an approximation factor of $2 \cdot (\max _{j \in [m]} |S_j ^* |) $ if each $g_j$ satisfies the  singleton minimal property, where 
$\mathcal{S} ^* = (S_1 ^* ,\, \ldots ,\, S_m ^* )$ is an optimal partition. 
\end{prop}

A proof of Proposition \ref{th:simpleb} uses the definition of the initial partition $\mathcal{S} ^{(0)}$, the following simple relation 
$
M^{(0)} _j (S_j ^* ) 
= \sum _{i \in S_j ^* } g_j (\{ i \}) 
\le |S_j ^*| g_j (S_j ^* ) ,
$
and the approximation factor of 2 of the algorithm for \eqref{eq:milp}. 
The initial partition of \textsf{MMin} already attains the approximation bound. 

% For nondecreasing SALB, the approximation bound of Proposition \ref{th:simpleb} will be improved 
% in \S{}\ref{sec:ana}.

\subsection{Analysis of the approximation algorithm for nondecreasing SALB}
\label{sec:ana}

For nondecreasing SALB, the bound of Proposition \ref{th:simpleb} can be improved by using the curvatures. 
We give a worst-case approximation factor of the algorithm \textsf{MMin} for the nondecreasing SALB. 
The following result is a generalization of the result of \cite{WIWBB15} for the  nondecreasing SMLB. 

\begin{theorem} \label{th:ag}
The algorithm \textsf{MMin} for the nondecreasing subadditive load balancing achieves an approximation factor of $2 \cdot (\max _{j \in [m]} \frac{|S_j ^* |}{1+(|S_j ^* | -1 ) (1 - \kappa _g (S_j ^* ))}) $, where 
$\mathcal{S} ^* = (S_1 ^* ,\, \ldots ,\, S_m ^* )$ is an optimal partition, and $\kappa _{g_j} (S)$ is the curvature of $g_j $ at $S \subseteq V$. 
\end{theorem}

\noindent
We give a proof of Theorem \ref{th:ag} with the aid of the curvatures of a subadditive set function. 
\vspace{-2mm}

\paragraph{Curvatures and modular approximation functions.} 
As with the submodular case \cite{Von10,IJB13}, 
we consider the curvatures of subadditive functions. 
Suppose that $g$ is a normalized nondecreasing subadditive set function with $g (\{ i \}) > 0 $ for all $i \in V$.
Define the curvature $\kappa _{g} (S) $ of $g$ at $S \subseteq V$ 
as 
\begin{align*}
\kappa _{g} (S) 
&= 1 - \min _{A \subseteq S,\, i \in A} \frac{g(i \, | \, A \setminus \{ i \} )}{g(\{ i \} )}.
\end{align*}
Denote the total curvature $\kappa _{g} (V)$ by $\kappa _{g}$. 
% By the definition, for all $S' \subseteq S \subseteq V$ and $i \in S'$, we have
% $\frac{g(i \, | \, S' \setminus \{ i \} )}{g(\{ i \} )} 
% \ge 1 -  \kappa _g (S) $
% and thus 
% \begin{align} \label{eq:cur-1}
% g(i \, |\, S' \setminus \{ i \} ) \ge (1 - \kappa _g (S) ) g(\{ i \}).
% \end{align}
%
% The following lemma evaluates to what extent the modular function $M(S ) = \sum _{i \in S} g(\{ i \}) \ (S \subseteq V)$ closely approximates $g$. 

% A set function $\widehat{g}: 2^V \to \mathbb{R}$ is an $\gamma$-approximation of $g$ if
%  $g(S) \le \widehat{g} (S) \le \gamma g(S)$ for all $S \subseteq V$.
% The following lemma evaluates to what extent the modular function $M(S ) = \sum _{i \in S} g(\{ i \}) \ (S \subseteq V)$ closely approximates $g$.

By the definition of the curvatures, for all $S' \subseteq S \subseteq V$ and $i \in S'$, we have
$\frac{g(i \, | \, S' \setminus \{ i \} )}{g(\{ i \} )} 
%\ge \min _{A \subseteq S,\, j \in A} \frac{g(j \, | \, A \setminus \{ j \} )}{g(\{ j \} )} = 
\ge 1 -  \kappa _g (S) $
and thus
\begin{align} \label{eq:cur-1}
g(i \, |\, S' \setminus \{ i \} ) \ge (1 - \kappa _g (S) ) g(\{ i \}). 
\end{align}
A set function $\widehat{g}: 2^V \to \mathbb{R}$ is an $\gamma$-approximation of $g$ if
 $g(S) \le \widehat{g} (S) \le \gamma g(S)$ for all $S \subseteq V$.
The following lemma evaluates to what extent the modular function 
\[
M(S ) = \sum _{i \in S} g(\{ i \}) \ (S \subseteq V)
\]
closely approximates $g$.
\vspace{2mm}

\begin{lemma} \label{lem:app2}
If $0< \kappa _g <1$, it holds that
\[
g(S) \le \sum _{i \in S} g(\{ i \}) \le 
\frac{1}{1 - \kappa _g } 
g(S) \ \ (S \subseteq V).
\]
\end{lemma}

\begin{proof}
The first inequality directly follows from the subadditivity.
Suppose $|S| = h$ and $S = \{ i _1 ,\, \ldots ,\, i_h \}$. 
Let $S_k = \{ i _1 ,\, \ldots ,\, i_k \}$ for each $k=1, \ldots ,\, h$. 
Then, by using the inequality \eqref{eq:cur-1}, we have
\[
g(S) = \sum\limits _{k=1} ^h g(i_k \, | \, S_{k} \setminus \{ i_k \}) \ge (1 -  \kappa _g ) \sum\limits _{i \in S} g(\{ i \}) ,
\]
which shows the second inequality. 
%\vspace{-2mm}
\end{proof}

\noindent
By a more detailed non-uniform analysis, we can obtain a non-uniform version of Lemma \ref{lem:app2}.
\vspace{2mm}

\begin{lemma} \label{lem:app3}
For each $S \subseteq V$,  
if $0< \kappa _g (S) <1$, it holds that
\[
g(S) \le \sum _{i \in S} g(\{ i \}) \le 
\frac{|S|}{1 + (|S| - 1) (1-\kappa _g (S) ) }
% \frac{1}{1 - \kappa _g (S) } 
g(S)
.
\]
% \ \ (S \subseteq V).
\end{lemma}

\begin{proof}
Let $S \subseteq V$ be a subset with $|S| = m$. 
Fix an element $i \in S$. 
Then, we let $S = \{ j_1 ,\, j_2 ,\, \ldots ,\, j_m \}$ with $j_1 =i$, and set 
$S_0 = \emptyset $, $S_k = \{ j_1 ,\, \ldots ,\, j_k \} \ (k=1, \ldots ,\, m)$. 
By using the inequality (A1), we have 
\begin{align*}
g(S) - g(\{ i \} ) 
 = \sum ^m _{k=2} g(j_k | S_{k-1}) 
\ge (1-\kappa _g (S) ) \sum _{j \in S \setminus \{ i \} } g(\{ j \}) . 
\end{align*}
By summing up these inequalities for all $i \in S$, we have 
\begin{align*}
& |S| g(S) -  \sum _{i \in S} g(\{ i \} ) 
\ge 
(1-\kappa _g (S) ) (|S| - 1) \sum _{i \in S} g(\{ i \}) ,\\
%%%
& \Longrightarrow \  
|S| g(S) 
\ge 
(1 + (1-\kappa _g (S) ) (|S| - 1) ) \sum _{i \in S} g(\{ i \}) ,\\
%%%
& \Longrightarrow \  
\frac{|S|}{1 + (|S| - 1) (1-\kappa _g (S) ) } g(S) 
\ge 
\sum _{i \in S} g(\{ i \}) .
\end{align*}
\end{proof}

% \noindent
% By a more detailed non-uniform analysis, we can obtain, for all $S \subseteq V$,   
% \begin{align} \label{eq:app2}
% \sum _{i \in S} g(\{ i \}) \le \frac{|S|}{1+(|S| -1)(1 - \kappa _g (S)) } g(S) .
% \end{align}

\paragraph{Analysis of the approximation factor.} 
Here we assume that $g_1 ,\, \ldots ,\, g_m$ are nondecreasing, and $g_j (\{ i \}) > 0 $ for all $i \in V$ and $j\in [m]$. In addition, we use a $\gamma _{\mathrm{MLB}}$-approximation algorithm for problem \eqref{eq:mlb0} in Step 0 of the algorithm  \textsf{MMin}. 
Notice that the polynomial-time algorithm of \cite{LST90} for problem \eqref{eq:mlb0} achieves an approximation factor of 2, that is, $\gamma _{\mathrm{MLB}} = 2$. 

To prove Theorem \ref{th:ag}, we show that the initial partition $\mathcal{S} ^{(0)} $ of the algorithm  \textsf{MMin} attains the approximation factor. Thus, it suffices to show the following lemma.

\begin{lemma} \label{lem:ag0}
Let 
$\mathcal{S} ^{(0)} = (S_1 ^{(0)}, \, \ldots ,\, S_m ^{(0)})$ be an 
optimal partition of problem \eqref{eq:mlb0}, and let
$\mathcal{S} ^* = (S_1 ^* ,\, \ldots ,\, S_m ^* )$ be an optimal partition of the nondecreasing SALB.
Then, $\mathcal{S} ^{(0)}$ is a 
$(\max _{j \in [m]} \frac{|S_j ^* |}{1+(|S_j ^* | -1 ) (1 - \kappa _g (S_j ^* ))}) $-approximation solution of the nondecreasing SALB.
\end{lemma}

\begin{proof}
Let $\gamma _j ^* = \frac{|S^* _j |}{1+(|S^* _j | -1)(1 - \kappa _g (S^* _j )) }$ for each $j \in [m]$. 
In view of Lemma \ref{lem:app3}, we have   
$\sum _{i \in S^* _j } g_j (\{ i \}) \le \gamma ^* _j g_j (S^* _j )$ for each $j \in [m]$. Therefore, we have
\begin{align} 
\label{eq:ag0-1}
\max _{j \in [m]} \sum _{i \in S^* _j } g_j (\{ i \}) 
\le 
(\max _{j \in [m]} \gamma ^* _j)
 \cdot 
( \max _{j \in [m]} g_j (S^* _j ) ) . 
\end{align}
By the subadditivity and the optimality of $\mathcal{S} ^{(0)}$, we have
\begin{align}
\max _{j \in [m]} 
g_j (S^{(0)} _j ) 
\le 
\max _{j \in [m]} 
\sum _{i \in S^{(0)} _j } g_j (\{ i \}) 
\le \
\max _{j \in [m]} 
\sum _{i \in S^{*} _j } g_j (\{ i \}) .
\label{eq:ag0-2}
\end{align}
Combining \eqref{eq:ag0-1} and \eqref{eq:ag0-2},  
we can see that the partition $\mathcal{S} ^{(0)}$ is a 
$(\max _{j \in [m]} \gamma ^* _j ) $-approximation solution of the nondecreasing SALB.
\end{proof}

\subsection{Lower bound computation}
\label{sec:lb}

In order to evaluate the quality of the obtained approximation solution, a lower bound of the optimal value of SALB gives an estimate about how close the 
%approximation
 solution is to the optimal one. 
We introduce a method to compute a lower bound, which employs the idea behind one iteration of \textsf{MMin}. 
%The following lower bound computation method is a variant of one iteration of the algorithm \textsf{MMin}. 

\subsubsection{A general framework}
%\paragraph{A general framework.} 
Given an $m$-partition $\mathcal{S} ' = (S' _1, \, \ldots ,\, S'  _m) $ for SALB, and let $\alpha \ge 1$. 
Suppose that, for each $j = 1,\ldots ,\, m$, we can construct a modular set function $M_j : 2^V \to \mathbb{R}$ 
satisfying $\alpha g_j (S) \ge M_j (S) ,\, \forall S \subseteq V$ and $g_j (S' _j ) = M_j (S'_j ) $ (we call such $M_j$ an $\alpha$-approximate minorization set function of $g_j$ at $S'_j$). 
Then, an optimal solution $\mathcal{S} '' = (S_1 '' , \, \ldots ,\, S_m '' ) $ to the M-LB \eqref{eq:mlb} with $M_1, \ldots ,\, M_m$ provides a  lower bound of the optimal value of the SALB.
Let $\mathcal{S} ^* = (S^* _1, \, \ldots ,\, S^*  _m) $ be an optimal partition to the SALB. 
%, and let $\mathcal{S} '' = (S_1 '' , \, \ldots ,\, S_m '' ) $ be an optimal solution to the M-LB \eqref{eq:mlb} with $M_1, \ldots ,\, M_m$.
%
For each $j =1, \ldots ,\, m$, 
%the $\alpha$-approximate minorization set function $M_j$ of $g_j$ at $S'_j$ 
the function $M_j$ 
satisfies 
$M_j (S^* _j ) \le \alpha g_j (S^* _j)$. Thus, we have $\max _{j} M_j (S^* _j ) \le \alpha \max _{j} g_j (S^* _j)$. 
% Remember that $\mathcal{S} '' = (S_1 '' , \, \ldots ,\, S_m '' ) $ is an optimal solution to the M-LB (4) with $M_1, \ldots ,\, M_m$.
By the definition of $\mathcal{S} ''$, we have $\max _{j} M_j (S'' _j ) \le \max _{j} M_j (S^* _j )$. Combining these inequalities, we obtain 
\[
\max _{j=1, \, \ldots \, , m } g_j (S_j ^* ) \ge \frac{1}{\alpha } \max _{j=1, \, \ldots \, , m } M_j (S_j '' ) . 
\]
Therefore, 
% and thus, $LB$ is the lower bound of the optimal value.
% Then, we compute an optimal solution  $\mathcal{S} '' = (S_1 '' , \, \ldots ,\, S_m '' ) $ to the M-LB \eqref{eq:mlb} with $M_1, \ldots ,\, M_m$. 
\begin{align} \label{eq:lb}
LB = \frac{1}{\alpha } \max _{j=1, \, \ldots \, , m } M_j (S_j '' )
\end{align}
is a lower bound of the optimal value of the SALB. 

The question 
% lower bound computation
is how to construct $\alpha$-approximate minorization set functions with small $\alpha \ (\ge 1)$. 
In the remaining part, 
we deal with  
% the construction of modular approximation set functions in 
%the case of
% the SALB with minimum spanning tree functions (\S{}\ref{sec:example}). The supplementary material deals with the nondecreasing SMLB.
 the SALB with minimum spanning tree functions (\S{}\ref{sec:example}) and the nondecreasing SMLB.

\subsubsection{Lower bound for SALB with MST functions} \label{sec:lb_mst} 
% \paragraph{Lower bound for SALB with MST functions. } 
Given a minimum spanning tree function $MST : 2^V \to \mathbb{R}$ on the node set $V$ with root $r$ and a subset $S' \subseteq V$, 
we consider a construction of an $\alpha _{\mathrm{T}}$-approximate minorization set function $M_{\mathrm{T}}$ of $MST$ at $S'$ for some $\alpha _{\mathrm{T}} \ge 1$. 
To construct $M_{\mathrm{T}}$, we use the cost-sharing method for the minimum spanning tree game \cite{Bird76}. 

Let $\mathcal{T} _r $ be a minimum spanning tree w.r.t. $V \cup \{ r \}$.  
For each $i \in V$, let $p_i$ be the (unique) parent node of $i$ such that $i $ and $p_i$ are directly connected on the (unique) path from $i$ to $r$ in $\mathcal{T} _r $. 
The algorithm of Bird \cite{Bird76} sets the weight $w_i $ of $i \in V$ as the distance from $i$ to $p_i$. 
It holds that 
$MST (S) \ge \sum _{i \in S } w_i \ (\forall S \subseteq V)$ and $\sum _{i \in V } w_i = MST(V)$ 
since the weight vector $\boldsymbol{w} = (w_i ) _{i \in V} \in \mathbb{R} ^n $ is a \textit{core} of the minimum spanning tree game. 
We set $\alpha _{\mathrm{T}} := MST(S') / \sum _{i \in S' } w_i \ (\ge 1)$, and define the modular set function $M_{\mathrm{T}} : 2^V \to \mathbb{R}$ 
as $M_{\mathrm{T}} (S) = \alpha _{\mathrm{T}} \sum _{i \in S} w_i \ (S \subseteq V)$.
Then, we have 
$\alpha _{\mathrm{T}} MST (S) \ge M_{\mathrm{T}} (S) \ (\forall S \subseteq V)$ and $MST (S' ) = M_{\mathrm{T}} (S' )$. 
Therefore, the function $M_{\mathrm{T}}$ is an $\alpha _{\mathrm{T}}$-approximate minorization set function of $MST$ at $S'$. 

%Let us see the evaluation of the lower bound for the SALB with MST functions $MST_1, \ldots ,\, MST_m$.
Now we establish a lower bound for the SALB with MST functions $MST_1, \ldots ,\, MST_m$. 
Given an $m$-partition $\mathcal{S} ' = (S' _1, \, \ldots ,\, S'  _m) $,  
we compute weight vector $\boldsymbol{w} ^{(j)} $ for each $MST_j$, and 
we define 
\begin{align} \label{eq:amax}
\alpha _{\mathrm{max} } = \max _{j = 1,\ldots , m} \alpha _j , 
\end{align}
where $\alpha _j = MST_j (S' _j) / \sum _{i \in S' _j} w_i ^{(j)} $. Then, we can obtain the lower bound $LB$ in \eqref{eq:lb} with $\alpha = \alpha _{\mathrm{max }}$. 
The quality of $LB$ is empirically evaluated in Section 6.

% Let $R = \{ r_1, \ldots ,\, r_m \} $ be a set of root nodes and $V = \{ 1,\, \ldots ,\, n \}$ be a set of other nodes. 
% For each $r_j $, let $MST _j : 2^V \to \mathbb{R}$ be a minimum spanning tree function on $V$ with root $r_j$.\

\paragraph{MST case with uniform node weights. } 
Given a uniform node weight $\beta \ge 0$, 
we can replace $MST_1, \ldots ,\, MST_m$ with $MST_1 ^{\beta }, \ldots ,\, MST_m ^{\beta }$, where 
$MST_j ^{\beta } (S) = MST_j (S) + \beta |S| \ (\forall S \subseteq V) $. 
Also in this case, the lower bound for SALB can be computed in the same way as the MST case.
We let 
\begin{align} \label{eq:amaxbeta}
\alpha _{\mathrm{max} } ^{\beta } = \max _{j = 1,\ldots , m} \alpha _j ^{\beta }, 
\end{align}
where $\alpha _j ^{\beta }  = MST_j ^{\beta } (S' _j) / \sum _{i \in S' _j} (w_i ^{(j)} + \beta )$.
Then, we can obtain $LB$ in \eqref{eq:lb} with $\alpha = \alpha _{\mathrm{max }} ^{\beta }$.

\subsubsection{Lower bound for nondecreasing SMLB} \label{sec:lb_sm} 

Given a normalized nondecreasing submodular set function $f : 2^V \to \mathbb{R}$ and a subset $S' \subseteq V$, 
we show a way of constructing an $\alpha$-approximate minorization set function $M_f$ of $f$ at $S'$ with $\alpha =1$.  

Let $\boldsymbol{w} = (w_i ) _{i \in V} \in \mathbb{R} ^n $ be an optimal solution to the linear optimization problem 
$\max _{\boldsymbol{z} \in \mathrm{P} (f ) } 
\langle \boldsymbol{I} _{S' } , \boldsymbol{z} \rangle $
over the polymatroid 
$\mathrm{P} (f )= 
\{ \boldsymbol{z} \in \mathbb{R} ^n : \mbox{$z(S) \le f(S)$} \ (\forall S \subseteq V) \} 
\cap \mathbb{R} _{\ge 0} ^n$
%(\S{}\ref{sec:si})
, where $\boldsymbol{I} _{S'} \in \{0,\, 1\} ^n$ is the characteristic vector of $S'$. 
The vector $\boldsymbol{w}$ can be computed efficiently via the greedy algorithm of Edmonds \cite{Edm70}.
Define the modular set function 
$M_f : 2^V \to \mathbb{R}$ as 
$M_f (S) = \sum _{i \in S} w_i \ (S \subseteq V)$.
% \begin{align} \label{eq:m_f} 
% M_f (S) = \sum _{i \in S} w_i \ (S \subseteq V).
% \end{align} 
Then, the definition of the polymatroid implies that $f (S) \ge M_f (S) ,\ \forall S \subseteq V$, and  
the correctness of the greedy algorithm of \cite{Edm70} implies that $f (S' ) = M_f (S' )$. 
Therefore, the function $M_f$ is an exact minorization set function of $f$ at $S'$.

\paragraph{The greedy algorithm of Edmonds. } 
To make the paper self-contained, we describe the greedy algorithm of Edmonds \cite{Edm70} 
%(1970)
 in detail. 
% Let $\boldsymbol{a} = (a_i ) _{i \in V} \in \mathbb{R} ^n _{\ge 0}$ be a nonnegative coefficient vector.  
For a nonnegative coefficient vector $\boldsymbol{a} = (a_i ) _{i \in V} \in \mathbb{R} ^n $, 
we deal with the linear optimization problem 
$\max _{\boldsymbol{z} \in \mathrm{P} (f ) } 
\langle \boldsymbol{a} , \boldsymbol{z} \rangle $
over the polymatroid $\mathrm{P} (f )$. 
Let $L= (v_1 ,\, v_2 ,\,  \ldots ,\, v_n ) $ be a total order of $V=\{ 1,\ldots ,\, n \}$ such that 
$a _{v_1 } \ge a _{v_2 } \ge \cdots \ge a_{v_n }$. 
Define $L(0) = \emptyset $, $L(1) = \{ v_1 \} $, $L(2) = \{ v_1 ,\, v_2 \} $,\,$\ldots$\,, and $L(n) = \{ v_1 ,\, \ldots ,\, v_n \} $. 
The greedy algorithm of Edmonds \cite{Edm70} sets $w_{v_h } := f(L(h)) - f(L(h -1) )$ for each $v_h \in V$. 
Then, the vector $\boldsymbol{w} = (w_i ) _{i \in V} \in \mathbb{R} ^n $ is known to be optimal to $\max _{\boldsymbol{z} \in \mathrm{P} (f ) } \langle \boldsymbol{a}, \boldsymbol{z} \rangle $. 
In addition, it holds that $f(L(h)) = w(L(h)) $ for each $h = 0, 1, \,\ldots,\, n$.  
In the case of the problem $\max _{\boldsymbol{z} \in \mathrm{P} (f ) } \langle \boldsymbol{I} _{S' } , \boldsymbol{z} \rangle $,
% in \S{}3.3.1, 
we fix any total order $L= (v_1 ,\, \ldots ,\, v_k ,\, v_{k+1},\,\ldots ,\, v_n) $ of $V=\{ 1,\ldots ,\, n \}$ such that $S'= \{ v_1, \ldots ,\, v_k \}$, where $k=|S'|$.

% \section{Intractability of subadditivity}
\section{Intractability of subadditivity, and countermeasures}
\label{sec:int}

Theorem \ref{th:ag} in \S{}\ref{sec:ana} shows a tractable aspect of subadditivity.
This section provides some intractable aspects of subadditivity.
In addition, as an alternative to the curvature of a subadditive set function,
 we introduce a concept of a pseudo-curvature.

% We show that the unconstrained minimization and the curvature computation are NP-hard for subadditive set functions, while they are polynomially solvable for submodular set functions. 

% \subsection{Unconstrained minimization}
% \label{sec:int_min}

\subsection{Intractability of subadditivity}

\paragraph{Unconstrained minimization.} 
For submodular set function $f : 2^V \to \mathbb{R}$ with $V = [n]$, the unconstrained minimization problem $\min _{S \subseteq V} f(S) $ is exactly solved in polynomial time \cite{GLS88,Sch00,IFF01}. 
On the other hand, we prove that the unconstrained subadditive minimization is not tractable. 
We derive the intractability from the NP-hardness of the prize-collecting Steiner tree (PCST) problem of Goemans and Williamson \cite{GW95}.

\begin{theorem} \label{th:min}
For subadditive set function $g : 2^V \to \mathbb{R}$, the problem $\min _{S \subseteq V} g(S) $ is NP-hard.
\end{theorem}

\begin{proof}
Given an $n$-dimensional nonnegative prize vector $\boldsymbol{p} = (p_i ) _{i \in V } \in \mathbb{R} ^n$ and a minimum spanning tree function $MST :2^V \to \mathbb{R}$ defined in \S{}2.2, 
let us consider a set function $PCST: 2^V \to \mathbb{R}$ defined by 
\begin{align*}
\label{eq:pcst_func}
PCST (S) = MST(S) + p(V \setminus S) \ \ (S \subseteq V).
\end{align*}
The PCST of Goemans and Williamson \cite{GW95}, which is NP-hard, coincides with the minimization problem $\min _{S \subseteq V} PCST(S) $. 
In addition, owing to the subadditivity of $MST(S)$ and the nonnegative modularity of $p(V \setminus S)$, 
$PCST$ is subadditive. 
Therefore, the PCST is a special case of the unconstrained subadditive minimization problem. 
Thus, the subadditive function minimization is NP-hard.
\end{proof}

% \subsection{Curvature computation}
% \label{sec:int_cur}

\paragraph{Curvature computation.} 
For a submodular set function, 
it is easy to calculate curvatures \cite{Von10,IJB13}. 
On the other hand, we prove that the curvature computation is not a trivial task in the subadditive case. 
We derive the intractability from the NP-hardness of the maximization problem of a nondecreasing submodular set function, 
which is a well-known NP-hard problem. 
%which generalizes the well-known Max-Cut.

\begin{theorem} \label{th:cur}
For a subadditive set function $g : 2^V \to \mathbb{R}$ and $S \subseteq V$, 
the computation of curvature 
$\kappa _{g} (S)  = 1 - \min\limits _{A \subseteq S,\, i \in A} 
\frac{g(A) - g(A \setminus \{ i \} )}{g(\{ i \} )} $
 is NP-hard.
\end{theorem}

\vspace{-5mm}

\begin{proof} 
Let $i ^*=n$ be a fixed element of $V = [n] $, and let $U= V \setminus \{ i^* \} = [n-1]$. 
In view of the definition of the curvature, it suffices to show the NP-hardness of the minimization problem 
$\min _{A \subseteq U} (g(A \cup \{ i ^* \} ) - g(A )) $. 
To prove the NP-hardness, we construct a subadditive function $g $ using a nonnegative submodular function.

Let $\widetilde{f} : 2^U \to \mathbb{R}$ be a general nonnegative submodular function. 
Then, we define a set function $g: 2^V \to \mathbb{R} $ with the ground set $V= U \cup \{ i^* \}$ as
\begin{align*}
g(A) = 
\left\{
\begin{array}{ll}
\widetilde{f} (A)& (A \subseteq U = V \setminus \{ i^* \}), \\
0& (i^* \in A \subseteq V). 
\end{array}
\right.
\end{align*}
Trivially, $g$ is nonnegative. Moreover, we can see that $g$ always satisfies the subadditive inequality 
$g(A) + g(B) \ge g(A \cup B) $ for all $A,\, B \subseteq V$.
((i) If $i^* \in A \cup B$, we have $g(A \cup B) = 0 \le g(A) + g(B)$.
(ii) If $i^* \notin A \cup B$, the nonnegative submodularity of $\widetilde{f}$ gives the subadditive inequality of $g$.)
Now,  the minimization problem $\min _{A \subseteq U} (g(A \cup \{ i ^* \} ) - g(A )) $ is equivalent to 
the submodular maximization problem $\max _{A \subseteq U} \widetilde{f} (A)$, which is known to be NP-hard. 
Therefore, the curvature computation is NP-hard for a subadditive set function. 
\end{proof}

\subsection{Pseudo-curvatures}

As a countermeasure to the difficulty of the curvature calculation of a subadditive set fuction (Theorem \ref{th:cur}), we introduce a concept of a pseudo-curvature.

Given a nonnegative subadditive set function $g: 2^V \to \mathbb{R}$,  
we suppose that $g$ can be decomposed as follows:
\[
g(S) = g_{+ } (S) + f_{+ } (S) \ (\forall S \subseteq V), 
\]
where $g_{+ } $ is subadditive and approximately or exactly nondecreasing, and $f_{+ } $ is nonnegative and submodular (or nonnegative and modular). 
If $g_{+ } $ is exactly nondecreasing, 
the total curvature $\kappa _g$ of $g$ can be bounded as follows:
\begin{align*}
\kappa _g 
&= 
\mbox{$
1 - \min\limits _{A \subseteq V,\, i \in A} 
\frac{g_+ (i \, | \, A \setminus \{ i \} ) + f_+ (i \, | \, A \setminus \{ i \} ) }{g(\{ i \} )} 
$}
\\
&\le 
\mbox{$
1 - \min\limits _{A \subseteq V,\, i \in A} 
\frac{f_+ (i \, | \, A \setminus \{ i \} ) }{g(\{ i \} )}
$}
\\
& %\le 
= 
\mbox{$  
1 - \min\limits _{i \in V} \frac{f_+ (i \, | \, V \setminus \{ i \} ) }{g(\{ i \} )} 
$} .
\end{align*}
An appropriate decomposition $g = g_{+ } + f_{+ }$ would make the value $\hat{\kappa } _g :=
 1 - \min _{i \in V} \frac{f_+ (V ) - f_+(V \setminus \{ i \} ) }{g(\{ i \} )} $ 
 a reasonable alternative to the total curvature $\kappa _g$ even if $g_{+ } $ is approximately nondecreasing.
We call $\hat{\kappa } _g$ a pseudo-curvature.

Let us consider the case where $g(S) = MST ^{\beta } (S) = MST (S) + \beta |S| $ defined in \eqref{eq:mstbeta}. The function $MST$ is nonnegative and subadditive (Lemma \ref{lem:mst}). 
Although $MST$ is not strictly nondecreasing, it can be regarded as an approximately nondecreasing set function. 
In addition, the term of $\beta |S| (= f_+ (S))$ is modular and nonnegative. 
Therefore, the pseudo-curvature $\hat{\kappa } _{MST ^{\beta }} $ is given by  
\begin{align} 
\hat{\kappa } _{MST ^{\beta }} 
 &= 1 - 
\mbox{$\min\limits _{i \in V} \frac{\beta |V| - \beta |V \setminus \{ i \} |}{MST^{\beta } (\{ i \}) } $}
\notag
\\
 &= 1 - 
\mbox{$\min\limits _{i \in V} \frac{\beta }{d ( r,  i ) + \beta } $}.
\label{eq:pcmstbeta}
\end{align}
The relationship between the performance of the proposed algorithm and the pseudo-curvature will be discussed in 
Section \ref{sec:er}.

\section{Application to multi-robot routing}
\label{sec:mrr}

This section explains an application of the subadditive load balancing (SALB) to the multi-robot routing (MRR) problem with the minimax team objective.

% \subsubsection*{Optimizing the multi-robot routing}

For a set of robots $\mathcal{R} = \{ r_1 ,\ldots , r_m \}$, a set of targets, $\mathcal{T} = \{ t_1 ,\ldots , t_n \}$, and any $i,\, j \in \mathcal{R} \cup \mathcal{T}$, a nonnegative cost (distance) $d(i,j) \ge 0 $ is determined. The cost function $d: (\mathcal{R} \cup \mathcal{T}) \times (\mathcal{R} \cup \mathcal{T} ) \to \mathbb{R}$ is symmetric and satisfies triangle inequalities. 
We consider an allocation of targets to robots. Let $S_j \subseteq \mathcal{T}$ be a target subset allocated to robot $r_j \in \mathcal{R}$. 
MRR with the minimax team objective asks for finding an $m$-partition $\mathcal{S} = ( S_1 ,\, \ldots ,\, S_m ) $ of $\mathcal{T}$ and a path $P_j$ for each robot $r_j \in \mathcal{R}$ that visits all targets in $S_j $ 
%KISHI: No need to return to the original position
%and returns to the original position of $r_j$ 
so that a team objective is optimized \cite{MRR05} as follows:   
%
%The minimax team objectives is defined as follows:
\begin{align*} 
% \mathrm{Minisum:}& \ 
% \min _{\mathcal{S} } \sum _{j \in \mathcal{R}} RPC _j (S_j )\\
% \ \mbox{ or } \ 
% \min _{\mathcal{S} } \sum _{j \in \mathcal{R}} RTC _j (S_j ), 
% \\
  \mathrm{Minimax:} \ \ & \min _{\mathcal{S} } \max _{j \in \mathcal{R}} RPC _j (S_j )
  \\
& \ \mbox{or } \  \min _{\mathcal{S} } \max _{j \in \mathcal{R}} RTC _j (S_j ), 
\end{align*}
where $RPC _j( S_j )$ is the minimum value of the robot path cost (RPC) for robot $r_j \in \mathcal{R}$ to visit all targets in $S_j$, and $RTC _j ( S_j )$ is the minimum value of the robot tree cost (RTC) for robot $r_j \in \mathcal{R}$, i.e., $RTC _j( S_j ) $ is the sum of edge cost of 
an MST on $\{ r_j \} \cup S_j$. 
%
%In order to uniformly balance the work among the robots, considering the minimax team objective is often necessary in practice, since 
%the minisum tends to assign most targets to only a very small number of robots and generate overworked and underworked robots \cite{KN16}. 
%
The RPC is intractable due to the NP-hardness of the traveling salesperson problem, but the RTC is tractable. 
An MST on $\{ r_j \} \cup S_j$ can be converted to a robot path on $\{ r_j \} \cup S_j$ whose cost is within $1.5 \cdot RPC _j( S_j ) $ (see, e.g., \cite{Vaz01}). 

Approximation algorithms especially based on sequential single-item auctions have been extensively studied to solve MRR \cite{KKT10}. 
\begin{comment}
This approach regards robots and targets as bidders and auction items, respectively. 
In each round where only one target is assigned to only one winner,
each robot attempts to obtain an unassigned target with a bid score 
%sends the server a bid score on an unassigned target that robot attempts to obtain.
%The bid score is
calculated by using the distance information on the set of targets currently assigned to the robot plus the unassigned target the robot bids on. After receiving the bids from all the robots, the server decides the winner (robot) that can actually acquire the target. This step is repeated until all the targets are allocated. % to the robots. 
\end{comment}
Our approach based on SALB provides a new, different way of tackling MRR with the minimax team objective.    
% possibly opening up a new research opportunity to develop new subadditive load balancing algorithms. 
Because $RTC _j ( S_j ) $ is an MST function (see \S{}\ref{sec:example}), it is an example of the subadditive set function. 
In addition, the lower bound analysis of \S{}\ref{sec:lb} can be utilized. 

We can also deal with the processing time of targets.
Let $\beta \ge 0$ be a (uniform) waiting time it takes each robot to process each target.
Then, each $RTC _j $ becomes an MST function with a uniform node weight, which is defined in \eqref{eq:mstbeta}. 
Furthermore, even in this case, we can use the lower bound analysis of \S{}\ref{sec:lb}, and the pseudo-curvatures can be computed in view of \eqref{eq:pcmstbeta}.

\section{Experimental results}
\label{sec:er}

\begin{table*}[t] 
\begin{center}
\ \vspace{-6mm}
\caption{Quality of RTC for {\amst} and {\modumin}}\label{tab:mst}
\hspace{-3mm}
{\scriptsize
% \footnotesize
\begin{tabular}{|c||c|c|c|c|c|c|c|c||c|c|c|}
\hline
Target &  \multicolumn{2}{c|}{\amst} & \multicolumn{2}{c|}{Initial}  & \multicolumn{2}{c|}{{\modumin}} & \multicolumn{2}{c||}{{\modumin} + {\amst}} & \multicolumn{3}{c|}{Lower bound}\\
size & RTC       & Time (s)  & RTC       & Time (s) & RTC       & Time (s) & RTC       & Time (s) & RTC & Time (s) & $\alpha_{\mathrm{max}}$\\
\hline
\hline
50         & 4220 & 0.00024      & 4735  & 0.78 & 3827 & 1.32    & {\bf 3628} & 0.52 & 1914 & 193.26 & 1.66\\
\hline
100        & 5484 & 0.0016        & 6626 & 15.50 & 5111 & 16.43   & {\bf 4797} & 0.73 & 3064 & 13.32 & 1.42\\
\hline
120        & 5886 & 0.0042        & 7147  & 81.48 & 5474  & 82.54  & {\bf 5189} & 0.87  & 3422 & 7.17 & 1.38\\ 
\hline
\end{tabular}
}
\end{center}
\ \vspace{-6mm} 
\end{table*}

We performed empirical evaluation in the MRR domain on a machine with four Intel CPU cores
(i5-6300U at 2.40GHz, only one core in use) and 7.4 GB of RAM. 
%We conducted the performance evaluation in the MRR domain on a machine consisting of four Intel CPU cores
%(i5-6300U at 2.40GHz) and 7.4 GB of RAM with only one core in use. 
%%
%\subsection{Setups}
%We implemented the modularization-minimization algorithm in C++, denoted by {\modumin}.
Our C++ implementation of the modularization-minimization algorithm denoted by {\modumin}, runs IBM ILOG CPLEX to optimally solve each LP problem generated
by {\modumin}. 
{\modumin} terminates if the value of the LP problem in the current iteration agrees with one in a previous iteration. 
%$\max_j g(S_j^{(k)}) = \max_j g(S_j^{(k')})$ for some $k' \leq k$. 
%This termination condition may not be the best, but our empirical evaluation shows it is sufficient to significantly improve the initial partitions.
%
We also implemented the following well-known MRR algorithms: 
\vspace{2mm}\\
\begin{tabularx}{166mm}{lX} 
$\bullet $ \hspace{-3mm} & 
{\amst} \cite{MRR05} is a standard algorithm presented in Algorithm Greedy, which is very similar to \textsc{GreedMin} of Wei \textit{et al.} \cite{WIWBB15}. 
%We used both RTC and RPC for performance comparison. 
When {\amst} calculates the RPC value, each robot's MST needs to be converted to a path. 
We used short-cutting \cite{LLKS85}, commonly used in MRR \cite{MRR05,KS08}. 
\vspace{1mm}\\
$\bullet $ \hspace{-3mm} & 
{\apath} \cite{MRR05} is a standard auction algorithm in the literature \cite{KKT10}. 
Starting with a null path, each robot greedily extends its path with the insertion heuristic \cite{LLKS85}. The robot with the smallest RPC wins an unassigned target in each round. This step is repeated until all targets are assigned.
%Since {\apath} does not computes RTC, we used the RPC metric for comparison. 
\end{tabularx}\vspace{2mm}\\
%
% \begin{itemize}
% \setlength{\itemsep}{5pt} 
% \setlength{\parskip}{0pt} 
% \setlength{\itemindent}{0pt} 
% \setlength{\labelsep}{5pt} 
% \item {\amst} \cite{MRR05} is a greedy algorithm presented in Algorithm Greedy.\footnote{Wei \textit{et al.} \cite{WIWBB15} also present a greedy algorithm called \textsc{GreedMin} that is very similar to this algorithm.}
% We used both RTC and RPC for performance comparison. 
% When {\amst} or {\modumin} calculates the RPC value, the MST of each robot needs to be converted into a path. 
% We used so-called short-cutting of TSP \cite{LLKS85}, which is also common in the MRR literature, e.g., \cite{MRR05,KS08}. 
% \item {\apath} \cite{MRR05} is a standard auction-based algorithm in the literature \cite{KKT10}. 
% Starting with a null path, each robot greedily attempts to extend its path based on the TSP insertion heuristic \cite{LLKS85}. In each iteration (bidding phase), the robot that has the smallest (suboptimal) RPC obtains an unassigned target. This step is repeated until all targets are assigned. Since {\apath} directly computes RPC, we used the RPC metric for performance comparison. 
% \end{itemize}
%
In calculating the RPC value for {\modumin}, our approach generated a path based on the insertion heuristic with a partition calculated by {\modumin} (i.e., the assignment optimized for RTC). 
We prepared a road map of the Hakodate area in Japan and precomputed distances between locations. 
%by using OSRM \cite{OSRM}. %, which is a freely available, efficient implementation to compute travel distances on road maps. 
That is, the distance was immediately retrieved when necessary. This is a common experimental setting for MRR, e.g., \cite{KTZS07,ZKT06}. 
In practice, when the map only contains a small city, all distance information fits into memory
%to compute travel plans,
e.g., \cite{NSHSMKKN14}.
The precomputed distances correspond to driving times in second. 
%such as for multiple shared taxis \cite{NSHSMKKN14}. 
%, which inherits some essential properties behind MRR. 
%Future work when we need dynamic distance computation. 
%
We always set the robot size to five, but prepared three cases for the target size (50, 100 and 120). 
% to evaluate the scaling behavior of the algorithm. 
Each case consisted of 100 instances by randomly placing agents and targets.
%on the map.

%%%%%%%%%%%%%%%   old table   %%%%%%%%%%%%%%% 
% \begin{table*}[t] 
% \begin{center}
% \caption{Quality of RTC for {\amst} and {\modumin}}\label{tab:mst}
% {\small
% \begin{tabular}{|c||c|c|c|c|c|c|c|c|}
% \hline
%            &  \multicolumn{2}{c|}{\amst} & \multicolumn{2}{c|}{Initial}  & \multicolumn{2}{c|}{{\modumin}} & \multicolumn{2}{c|}{{\modumin} + {\amst}}\\
% \hline
% Target size & RTC       & Time (s)  & RTC       & Time (s) & RTC       & Time (s) & RTC       & Time (s)\\
% \hline
% \hline
% 50         & 42,208 & 0.0002      & 47,392  & 0.80 & 38,202 & 1.29    & 36,326 & 0.63 \\
% \hline
% 100        & 54,840 & 0.17        & 66,131 & 12.89 & 50,437 & 13.81   & 47,752 & 0.76\\
% \hline
% \end{tabular}
% }
% \end{center}
% \ \vspace{-6mm} 
% \end{table*}

% \begin{table*}[t] 
% \begin{center}
% \caption{Quality of RPC for each method}\label{tab:path}
% {\small
% \begin{tabular}{|c||c|c|c|c|}
% \hline
% Target size   &  {\amst} &  {\apath} & {\modumin} & {\modumin} + {\amst} \\ %& {\modumin} + {\apath} \\
% \hline
% \hline
% 50         & 52,336 & {\bf 49,792} & 50,984  &  50,446  \\ %& 51,897.0  \\
% \hline
% 100        & 72,553 & 69,002  & 68,282  & {\bf 68,067}  \\ %&  69,380.52\\
% \hline
% \end{tabular}
% }
% \end{center}
% \ \vspace{-4mm} 
% \end{table*}
%%%%%%%%%%%%%%%   END of old table   %%%%%%%%%%%%%%% 

%%%%%%%%%%%%%%%   new table   %%%%%%%%%%%%%%% 

\begin{table*}[t]
 \begin{center}
 \caption{Quality of RTC for {\amst} and {\modumin} with 120 targets and 5 robots.}\label{tab:wait}
 \hspace{-3mm}
 {\scriptsize
 \begin{tabular}{|c||c|c|c|c|c|c|c|c||c|}
 \hline
 Waiting &  \multicolumn{2}{c|}{\amst} & \multicolumn{2}{c|}{Initial}  & \multicolumn{2}{c|}{{\modumin}} & \multicolumn{2}{c||}{{\modumin} + {\amst}} & Pseudo \\
 Time & RTC       & Time (s) & RTC       & Time (s) & RTC       & Time (s) & RTC       & Time (s)  & Curvature \\
 %& RTC & Time & $\alpha_{\mathrm{max}}$\\
 \hline
 \hline
 10         & 6095 & 0.0032 & 7195 & 70.77  &  5617 & 72.11 & {\bf 5393} & 1.17 &  0.997\\
 \hline
 30        & 6656 & 0.0027 & 7668 &  21.33  & 6100 & 22.80 & {\bf 5907} & 1.35 & 0.991\\
 \hline
 60        & 7476 & 0.0025 & 8314 &  32.68  & 6790  & 34.49  & {\bf 6661} & 1.77 & 0.983\\
 \hline
\end{tabular}
 }
\end{center}
\ \vspace{-6mm}
% \end{table*}
%
% \begin{table*}[t] 
\begin{center}
\caption{Quality of RPC for each method}\label{tab:path}
{\scriptsize
%% \footnotesize
\begin{tabular}{|c||c|c|c|c|}
\hline
Target size   &  {\amst} &  {\apath} & {\modumin} & {\modumin} + {\amst} \\ %& {\modumin} + {\apath} \\
\hline
\hline
50         & 5233 & 4979 & 5011  &  {\bf 4767}  \\ %& 51,897.0  \\
\hline
100        & 7255 & 6900  & 7182  & {\bf 6827}  \\ %&  69,380.52\\
\hline
120        & 7756 & 7484  & 7870  & {\bf 7446}  \\ %&  69,380.52\\
\hline
\end{tabular}
}
\end{center}
\ \vspace{-6mm} 
\end{table*}
%%%%%%%%%%%%%%%   END of new table   %%%%%%%%%%%%%%% 

%\subsection{Results on RTC}

Table \ref{tab:mst} shows average RTC values and average runtimes. %in each scenario. 
We excluded {\apath} here, since it does not optimize for RTC. 
The values in ``Initial'' indicate the RTC values and runtimes for finding initial $m$-partitions of {\modumin}. 
The iterative procedure of {\modumin} is regarded as a step to further refine an initial solution. 
Therefore, another initial solution can be passed to this iteration, although the worst-case theoretical analysis of the solution quality remains future work. 
We prepared {\modumin} + {\amst} that performs the {\modumin} iterations but starts with an initial $m$-partition calculated by {\amst}. 

{\modumin} generated 9\%, 7\%, and 7\% smaller RTC values than {\amst} with 50, 100 and 120 targets, respectively.
While the RTC values of the initial partitions generated by solving the first LP problems were inferior to those
of {\amst}, {\modumin}'s iterative steps successfully improved the initial partitions. 

In case of 50 targets, {\modumin} performed 19 iterations on average, ranging between 4 and 76 iterations. 
% in the smallest and largest cases.
In case of 120 targets, the number of iterations was 22 on average, ranging between 6 and 92 iterations. 
%and the smallest and largest numbers of iterations were 6 and 92, respectively.
%
Among these iterations, {\modumin} typically spent 59-99 \% of the runtime in finding an initial partition.
The formulation of the initial LP problem is slightly different from those in the remaining iterations, which we hypothesize as a cause of the difficulty.  
%It was more difficult for CPLEX to solve this initial LP problem whose problem formulation is slightly different from the other LP problems 
%solved in the remaining iterations. 
Finding the initial partition tended to be harder with a larger target size. 
%This is a significant factor of the increase of the average runtime of {\modumin} with an increase of the target size (1.29 seconds with 50 targets versus 13.81 seconds with 100 targets). 
In solving the most difficult instance with 50 targets, {\modumin} needed 13.27 seconds to find an initial partition and only 0.75 seconds 
for the remaining iterations. 
%In the worst case with 100 targets, the runtimes for finding an initial partition and solving the remaining part were 365.87 seconds and 4.53 seconds, respectively.
There was one instance with 120 targets which took {\modumin} 7317 seconds to find an initial partition.   

{\modumin} + {\amst} performed best in terms of the RTC quality and runtime. It bypassed the significant overhead of the initial partition computation, and
started with a better initial RTC value than standard {\modumin}. 
%In addition, results show that the {\modumin} iterations successfully improved the performance of {\amst} further. 
This indicates that such a hybrid approach is important in practice. 
The results shown in \cite{MRR05} seem to indicate that the worst-case solutions of auction-based algorithms are bounded by $O(m)$. However, {\modumin} is a centralized approach, and 
a question remains open as future work regarding whether {\modumin} + {\amst} theoretically guarantees a better approximation factor in general or not.

We calculated an average lower bound and $\alpha_{\mathrm{max}}$ defined in \eqref{eq:amax} for each target size (see Table \ref{tab:mst}) with the method presented in \S{}\ref{sec:lb} and with the partition returned by {\modumin} + {\amst}.
These lower bounds indicate that on average the approximation factors of {\modumin} + {\amst} were empirically better than 1.90, and 1.57 and 1.52 for 50, 100, and 120 targets, respectively.
We observed that there were many instances whose larger lower bounds were returned
if different partitions (e.g., a partition that is better than in the previous iterations) were used for the lower bound computation. 
Therefore, in fact, {\modumin} + {\amst} must yield solutions closer to optimal than those potential approximation factors.

The average runtime to compute a lower bound increased when the target size decreased,
which was counter-intuitive. 
However, the potential approximation factors and the $\alpha_{\mathrm{max}}$ values increased with a decrease of the target size. Therefore, 
we hypothesize that the LP problem for lower bound calculation tends to be more difficult  
if the lower bound is farther than the optimal solution.

Next, let us see the performance of the algorithms with waiting time $\beta \ge 0$. 
Table \ref{tab:wait} shows the cases where waiting times are varied from 10--60 seconds with 5 robots and 120 targets.
The waiting time corresponds to the time necessary for a robot to collect and drop off a target. 
We observed similar tendencies, even when waiting times were introduced. {\modumin} returns better RTC values than {\amst},
but suffers from the computational overhead for solving the initial LP problem. 
{\modumin}+{\amst} bypasses this overhead as well as yields the best RTC values.
The values of the pseudo-curvature $\hat{\kappa } _{MST ^{\beta }}$ determined in \eqref{eq:pcmstbeta} close to 1 indicate a difficulty of performing theoretical analysis,
even with waiting times. When the waiting time $\beta $ was set to 60, the value of $\alpha^{\beta} _{\mathrm{max}}$ defined in \eqref{eq:amaxbeta} was 1.32,
which was smaller than that without waiting time (i.e., $\alpha_{\mathrm{max}}=1.38$ as shown in Table \ref{tab:mst}, and the pseudo-curvature is 1).
This result implies the relative tractability of the problem with 
%waiting time
  $\beta = 60$. 

%\subsection{Results on RPC}

%\paragraph{Results on RPC.} 
Table \ref{tab:path} shows RPC values of each method. 
%In addition to aforementioned {\modumin + \amst}, we tested {\modumin + \apath} which passes to {\modumin} the initial partition computed by {\apath}. 
Since our approach does not optimize for the RPC, 
algorithms that obtain better RTC values do not always yield better RPC values in theory.  
However, in practice, {\modumin} + {\amst} performed best of all methods. 
%yielding 4\% better solutions than {\apath} on 50, 100 and 120 targets, respectively. 

%{\modumin} performed better than {\amst} in the RPC metric, resulting in generating 3\% and 6\% better average RPC values with 50 and 100 targets, respectively. 

%While {\apath} performed best of all the methods with 50 targets, {\modumin} yielded better solutions than {\apath} with 100 agents. 
%A combination of {\modumin} and {\amst} slightly improved the performance over {\modumin}, resulting in performing best with the 100 targets. 

While {\apath} and {\amst} guarantee the same theoretical approximation factor to optimal RPC values in MRR, 
there has been consensus that {\apath} tends to perform better than {\amst} \cite{MRR05}, 
as is the case in our experiment. 
%As a result, most of the follow-up research was based on {\apath}, e.g.  \cite{KTZS07,ZKT06}. 
However, our results are important in the sense that they indicate that approaches that optimize solutions for the RTC metric (i.e., the MST function) and then convert
to a path have a potential to become a better approach than {\apath}, which would open up further research opportunities. 

%On the other hand, combining {\modumin} and {\amst} does not lead to an improvement to {\modumin}.

%The runtimes of {\modumin} + {\apath} and {\modumin} + {\amst} decrease significantly.

%Initial RTC for 44835.06 for {\modumin} + {\apath} with 50 agents \\
%Initial RTC for 58300.31 for {\modumin} + {\apath} with 100 agents

\section{Concluding remarks}

%We presented the modularization-minimization algorithm for the subadditive load balancing which is a generalization of the submodular load balancing, and gave an approximation guarantee for the nondecreasing subadditive case. 
We presented the modularization-minimization algorithm for the subadditive load balancing, and gave an approximation guarantee for the nondecreasing subadditive case. We also presented a lower bound computation technique for the problem. 
In addition, we evaluated the performance of our algorithm in the multi-robot routing domain. 
The application of subadditive optimization to AI are new, and 
subadditive approaches may open up a new field of AI and machine learning.
% Added some future work. Please modify as you like.
An example of future work is to elucidate the theoretical and empirical behaviors of the modularization-minimization algorithm with respect to the initial solutions. 
Our results about giving the MMin iteration procedure an initial partition calculated by
an MST-based greedy algorithm show the importance of the choice of the initial solutions. 
On other other hand, the question whether or not the iterative procedure currently contributes to improving the worst-case approximation factor remains unanswered.

% \newpage

{\small 
%\bibliography{slb_arxiv}

\bibliographystyle{plain}
}

% \newpage

\section*{Supplementary Material}

\section*{A.1 \ Proofs of Subadditivity}

We show the subadditivity of the minimum spanning tree function $MST$ 
and the facility location function $FL$
 defined in \S{}\ref{sec:example}. 
% In addition, we give the remaining proofs of the properties in Lemma 3.

%another example of non-submodular subadditive set function, a subadditive interpolation of a submodular set function.

\subsection*{A.1.1 \ Subadditivity of minimum spanning tree function} 

In the field of game theory, the subadditivity of the minimum spanning tree function is recognized in relation to the minimum spanning tree game (Bird \cite{Bird76}). 
Here, we describe the proof of Lemma \ref{lem:mst} to make the paper self-contained.
\vspace{1mm}

\noindent
\textbf{Lemma 1 (\S{}2.2.1).}  
\textit{
A minimum spanning tree function $MST : 2^V \to \mathbb{R}$ is nonnegative and subadditive.
}
\vspace{-1mm}

\begin{proof}
%\textit{Proof of Lemma 1.} 
By definition, nonnegativity is trivial.
For subsets $S, T\subseteq V$, 
let $E_S $ be the edge set of MST w.r.t. $\widetilde{S}$ and let $E_T $ be the edge set of MST w.r.t. $\widetilde{T}$. 
The graph $(S \cup T \cup \{ r \} ,\, E_S \cup E_T ) $ with a node set $S \cup T \cup \{ r \}$ and an edge set $E_S \cup E_T$ are connected. 
Thus we have 
$
MST(S) + MST(T) = \sum _{e \in E_S \cup E_T} d(e) \ge MST(S \cup T),
$
which shows the subadditivity of $MST$. 
\end{proof}
%\hfill $\square$\vspace{0mm}

\subsection*{A.1.2 \ Subadditivity of facility location function} 

We show the subadditivity of the facility location function.
\vspace{1mm}

\noindent
\textbf{Lemma 2 (\S{}2.2.1).}  
\textit{
A facility location function $FL : 2^V \to \mathbb{R}$ is nondecreasing and subadditive.}\vspace{-1mm}

\begin{proof}
We are given a set $V=\{1,\ldots ,\, n\} $ of customers and a finite set $F$ of possible locations for the facilities with 
opening costs $o_j \ge 0$ $(\forall j \in F)$ and connecting costs $c_{ij} \ge 0$ $(\forall (i,\, j) \in V \times F)$. 
For each edge subset $\mathcal{E} \subseteq V \times F$, let $N(\mathcal{E}) $ be a set of all endpoints of $\mathcal{E}$, and we denote 
$V(\mathcal{E}) := V \cap N(\mathcal{E}) $ and $F(\mathcal{E}) := F \cap N(\mathcal{E}) $. 

For $S \subseteq V$, $F' \subseteq F$, and $\mathcal{E} \subseteq V \times F$, 
we say that a triple $\mathcal{T} = (S, \, F' , \mathcal{E} ) $ is feasible if $S \subseteq V(\mathcal{E})$ and $F(\mathcal{E}) \subseteq F' $.
For any feasible triple $\mathcal{T} = (S, \, F' , \mathcal{E} ) $, we define $\mathrm{cost} (\mathcal{T} ):= \sum _{j \in F'} o_j + \sum _{(i,j) \in \mathcal{E}} c_{ij} $, and we have 
%For any feasible triple $\mathcal{T} = (S, \, F' , \mathcal{E} ) $, it holds that 
$FL(S) \le \mathrm{cost} (\mathcal{T} ) $.
In addition, for each $S \subseteq V$, there exists a feasible triple $\mathcal{T} ^* = (S, \, F^* , \mathcal{E} ^* ) $ such that 
$FL(S) = \mathrm{cost} (\mathcal{T} ^*)$. 

Suppose that $S' \subseteq S \subseteq V$. 
Let $\mathcal{T} ^* = (S, \, F^* , \mathcal{E} ^* ) $ be a feasible triple satisfying
$FL(S) = \mathrm{cost} (\mathcal{T} ^*)$. Then, $\mathcal{T} ' = (S' , \, F^* , \mathcal{E} ^* ) $ is also a feasible triple. 
Therefore, $FL(S' ) \le \mathrm{cost} (\mathcal{T} ') = \mathrm{cost} (\mathcal{T} ^*) = FL(S)$, which implies the nondecreasing property of $FL$. 

Suppose that $S_1, S_2 \subseteq V$. 
For each $k \in \{1,2\} $, let $\mathcal{T} ^* _k = (S _k , F^* _k , \mathcal{E} ^* _k ) $ be a feasible triple such that 
$FL(S_k ) = \mathrm{cost} (\mathcal{T} ^* _k)$. 
Since $\mathcal{T} ^* _1 \cup \mathcal{T} ^* _2 = 
(S _1 \cup S _2 , F^* _1 \cup F^* _2 , \mathcal{E} ^* _1 \cup \mathcal{E} ^* _2 ) $ is feasible, 
we have $FL(S _1 \cup S _2 ) \le \mathrm{cost} (\mathcal{T} ^* _1 \cup \mathcal{T} ^* _2 ) 
\le \mathrm{cost} (\mathcal{T} ^* _1 ) + \mathrm{cost} (\mathcal{T} ^* _2 ) = FL(S _1 ) + FL(S _2 )$,  
which implies the subadditivity of $FL$. 
\end{proof}

% the $FL : 2^V \to \mathbb{R}$ a . 

% Let $V=\{1,\ldots ,\, n\} $ be a set of customers, and let $F$ be a finite set of possible locations for the facilities. 
% For an edge subset $\mathcal{E} \subseteq V \times F$, let $N(\mathcal{E}) $ be a set of all endpoints of $\mathcal{E}$.  
% $V(\mathcal{E}) = V \cap N(\mathcal{E}) $, $F(\mathcal{E}) = F \cap N(\mathcal{E}) $. 

% For $S \subseteq V$, $F' \subseteq F$, and $\mathcal{E} \subseteq V \times F$. 

% , a certain service is provided by connecting customers to opened facilities. 

% Opening facility $j \in F$ incurs a fixed cost $o_j \ge 0$, and connecting customer $i \in V$ to facility $j \in F $ incurs a fixed cost $c_{ij} \ge 0$. 
% For any subset $S \subseteq V$ of the customers, let $FL(S) $ be the minimum cost of providing the service only to $S$. We call $FL : 2^V \to \mathbb{R}$ a facility location function. 

% We describe the proof of Lemma 2.

\section*{A.2 \ Supplementary Explanation for \S{}3.1} 
 
Let us see that it would be difficult to construct a majorization set function in the subadditive case.

Given a subadditive set function $g : 2^V \to \mathbb{R}$ and a subset $S' \subseteq V$, 
define the functions $M^1 $ and $M^2$ as, for all $S \subseteq V$, 
\begin{align*}
M^1 (S) = g(S' )
& \mbox{$+ \sum\limits _{i \in S \setminus S' } g( i \, | \, S' )$} \mbox{$ - \sum\limits _{i \in S' \setminus S } g( i \, | \, V \setminus \{ i \} ),$} \\
M^2 (S) = g(S' ) 
& \mbox{$  + \sum\limits _{i \in S \setminus S' } g( i \, | \, \emptyset )$} \mbox{$ - \sum\limits _{i \in S' \setminus S } g( i \, | \, S' \setminus \{ i \} )$}. 
\end{align*}
Both $M^1$ and $M^2 $ are majorization set functions of $g$ at $S'$ if $g$ is \textit{submodular} (Wei \textit{et al.} \cite{WIWBB15}).
%\cite{WIWBB15}.
%Let us see that these 
But these functions are not necessarily majorizing functions in the subadditive case. 

Let $g: 2^{\{ 1, 2, 3 \}} \to \mathbb{R}$ be the minimum spanning tree function in Figure 1 (b), which is nondecreasing, subadditive, and non-submodular. 
The function $M^1 $ with $S' = \{ 1 \}$ becomes
$M^1 (S) = 2+ \sum _{i \in S } a_i $ 
$(S \subseteq \{ 1,2,3 \} )$, where $a_1 =3$, $a_2 =a_3 =1$. $M^1$ does not majorize $g$ since $g(\{ 1,2,3\}) =9 $ and $M^1 (\{ 1,2,3\}) =7 $.  
The function $M^2 $ with $S' = \{ 1 , 2, 3\}$ becomes
$M^2 (S) = \sum _{i \in S } b_i $ 
$(S \subseteq \{ 1,2,3 \} )$, where $b_1=b_2 = b_3 =3$. $M^2$ does not majorize $g$ since $g(\{ 1\}) =5 $ and $M^2 (\{ 1\}) =3 $.

\end{document}